\pdfoutput=1

\documentclass[type=bachelor, language=english, oneside]{sjtuthesis}

\usepackage{ragged2e}
\usepackage{bbding}

\usepackage[backend=biber]{biblatex}
\usepackage{multirow}

\usepackage{dcolumn}
\newcolumntype{d}[1]{D{.}{.}{#1}}

\usepackage{longtable}

\usepackage{threeparttable}

\usepackage[ruled,vlined,linesnumbered]{algorithm2e}
\usepackage{listings}
\lstnewenvironment{codeblock}[1][]
  {\lstset{style=lstStyleCode,#1}}{}

\usepackage{siunitx}

\usepackage{ntheorem}
\usepackage{tikz}

\usepackage{hyperref}

\begin{document}

%TC:ignore

% 无编号内容

% 论文标题页
% \maketitle

% 原创性声明、版权授权页
% \originalitypage
% 如果不需要在原创性声明页面显示论文名，请使用 \originalitypage*
% \copyrightpage[scans/copyright.pdf]

% 使用罗马数字对前文编号
\frontmatter

% 摘要
% !TEX root = ../main.tex

\centering
\LARGE{\textbf{Exploring the Disproportion Between Scientific Productivity and Knowledge Amount}}
\vspace{5mm}

\newcommand\blfootnote[1]{%
\begingroup
\renewcommand\thefootnote{}\footnote{#1}%
\addtocounter{footnote}{-1}%
\endgroup
}

\large{
Luoyi Fu$^1$*, Huquan Kang$^2$*, Jianghao Wang$^3$, Ling Yao$^4$,\\ Xinbing Wang$^5$\Envelope, Chenghu Zhou$^6$\Envelope\blfootnote{
*\quad The two authors contribute equally to the work. \\ 
\Envelope\quad The two both are corresponding authors.  \\
\textbf{Email:} $^1$yiluofu@sjtu.edu.cn, $^2$kinghiqian@sjtu.edu.cn, $^3$wangjh@lreis.ac.cn, $^4$yaoling@lreis.ac.cn, $^5$xwang8@sjtu.edu.cn, $^6$zhouch@lreis.ac.cn \\
\textbf{Affiliation:} $^{125}$Shanghai Jiao Tong University, Shanghai 200240, China. $^{346}$State Key Laboratory of Resources and Environmental Information System, Institute of Geographic Sciences and Natural Resources Research, Chinese Academy of Sciences, Beijing 100101, China.\\ \\ 
\textbf{Note:} \\ During the exploration of this work, part of the results have been formed into a thesis. Therefore, this article is an expanded version of the thesis with more in-depth study, more supplementary aspects and more technical details.}}
\vspace{10mm}

\justifying
The pursuit of knowledge is the permanent goal of human beings. Scientists have developed numerous approaches to the representation of knowledge, and to extracting, discovering, learning, and reasoning about it. However, knowledge is often dynamic, going through human beings for knowing, invention, propagation, and problem-solving. Scientific literature, as the major medium that carries knowledge between scientists, exhibits explosive growth during the last century. Despite the frequent use of many tangible measures, such as citation\cite{citation1,citation2,citation3,citation4}, impact factor\cite{impactfactor} and g-index\cite{gindex}, to quantify the influence of papers from different perspectives based on scientific productivity, it has not yet been well understood how the relationship between scientific productivity and knowledge amount turns out to be\cite{assess1,assess2,assess3,assess4,assess5}, i.e., how the knowledge value of papers and knowledge amount vary with development of the discipline. This raises the question of whether high scientific productivity equals large knowledge amount. Here, building on rich literature on academic conferences and journals, we collect 185 million articles covering 19 disciplines published during 1970 to 2020, and establish citation network research area to represent the knowledge flow from the authors of the article being cited to the authors of the articles that cite it under each specific area. As a result, the structure formed during the evolution of each scientific area can implicitly tells how the knowledge flows between nodes and how it behaves as the number of literature (productivity) increases. By leveraging Structural entropy in structured high-dimensional space and Shannon entropy in unstructured probability space, we propose the Quantitative Index of Knowledge (KQI), which is taken as the subtraction between the two types of entropy, to reflect the extent of disorder difference (knowledge amount) caused by structure (order). With the aid of KQI, we find there exists significant disproportion between the growth of scientific productivity and knowledge amount. In other words, although the published literature shows an explosive growth, the amount of knowledge (KQI) contained in it obviously slows down, and there is a threshold after which the growth of knowledge accelerates. We demonstrate that these two observations can be explained by a BA model along with percolation theory, allowing us to prove quantitatively that polynomial literature publications will ultimately bring out linear knowledge growth. We find this phenomenon to be remarkably universal across diverse disciplines. And for a network of size $n$, the threshold of knowledge growth acceleration meets the average amount of references $m \sim a\log{n+1}$, where $a$, varying from discipline to discipline, represents the amount of related literature required in order to understand one specific article. This threshold reflects the sufficient preliminary study needed for a research area/paper/researcher to rise abruptly based on its accumulated strength, with a smaller value of $a$ meaning that the corresponding discipline is more inclined to produce disruptive works rather than developing ones. KQI helps researchers to construct the knowledge structure vein of each scientific research area, which is found to approximately follow the classic 80/20 rule with roughly 17\% selected extraordinary articles (with high KQI) set off by 83\% ordinary articles (with low KQI). Also, KQI is manifested to have power in digging out influential researchers/papers/institutions that might not be precisely portraited by those aforementioned quantity-based measures. These results not only deepen our quantitative understanding of the correlation between knowledge and scientific productivity, but meanwhile imply the importance of a dialectic viewing of the scientific articles with different values of KQI, i.e., those with low KQI might represent the repeatable, developing explorations that serve as the basis of scientific research, and the co-existence of both low and high KQI papers are indispensable for promoting a healthy academic ecological environment.

% 目录、插图索引、表格索引、算法索引
\tableofcontents
\listoffigures*
\listoftables*
\listofalgorithms*

% 主要符号对照表
% \input{contents/nomenclature}

%TC:endignore

% 使用阿拉伯数字对正文编号
\mainmatter

% 正文内容
% !TEX root = ../main.tex

\chapter{Introduction}

\section{Background}

With the growth of academic big data, the contradiction between the ability of human knowledge acquisition and the speed of information generation is increasingly prominent. Nowadays, academic literature has entered an explosive growth period with the further increase of scientific research investment. Statistics of conference papers and journal articles from 1800 to 2020 show that the number of literature today is three times that of 20 years ago and fifteen times that of 50 years ago (Fig. \ref{fig:Bibliometric}a). While the large volume of scientific papers might produce some ground-breaking knowledge, it also places researchers under the dilemma of reading fatigue. This predicament may be still sustainable for the newly emerging disciplines, but for disciplines that have undergone long-term development, the requirements for researchers to conduct research are demanding. By classifying literature into 292 sub-disciplines of concern to researchers, 39\% sub-disciplines have over 1 million literature, and 99\% have over 100 thousand literature, with which researchers are overwhelmed (Fig. \ref{fig:Bibliometric}b). Given the increasing interest in alleviating the burden of literature research for scientists, here we ask: Can we untangle the role of knowledge from productivity, and ease within the reading fatigue? To address this question, we try to quantify the knowledge amount of scientific productivity in multiple disciplines. 

\begin{figure}[!hbtp]
    \centering
    \subcaptionbox{Growth in the number of literature. The number of the literature showed an accelerated growth with the increase of years. The total number of literature in 2020 (198 million) has reached 3.3 times that of 2000 (61 million) and 15 times that of 1970 (13 million).}[6.4cm]{
        \includegraphics[height=5cm]{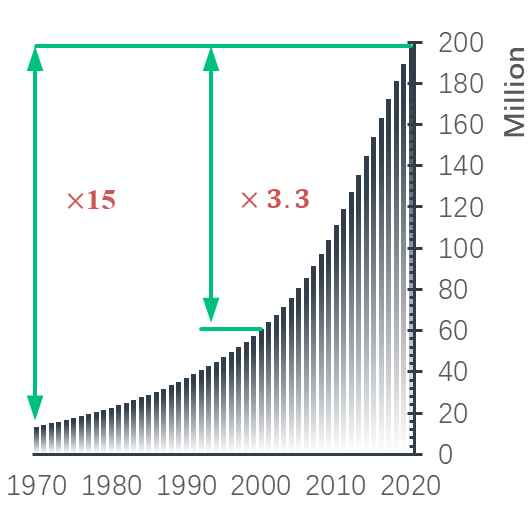}
    }
    \hspace{1cm}
    \subcaptionbox{Quantity distribution of literatures in different disciplines. Of all 292 second-level disciplines, more than 39\% have more than 1 million literatures by 2020 and almost all of them have more than 100 thousand literatures. }[6.4cm]{
        \includegraphics[height=5cm]{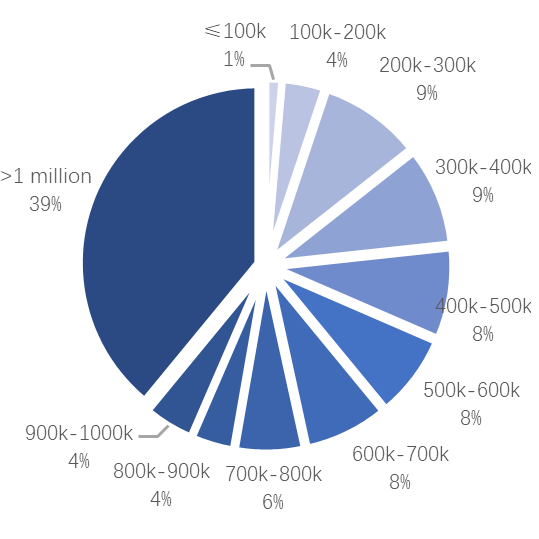}
    }
    \caption{Bibliometric statistics in the last fifty years.}
    \label{fig:Bibliometric}
\end{figure}

Productivity, representing the number of publications, or impact, often approximated by the number of citations a publication receives, are frequently used metrics to gauge a scientist's performance. Many quantitative indicators have been proposed around productivity. Citations\cite{citation1,citation2,citation3,citation4}, the epitome of scientific influence, and citation-based measures like Hirsch index\cite{hirsch}, g-index\cite{gindex}, impact factors\cite{impactfactor} and eigenfactors\cite{eigenfactor}, can help researchers to screen out influential literature from different dimensions respectively. However, those measures are simply statistical indicators based on the citation quantity that focuses on the portrayal of influence \cite{negativecitation}, falling short of reflecting how a knowledge inspires a new knowledge between different published scientific articles. Thus, we still lack a quantitative metric to reflect the knowledge contained in scientific productivity. Here we note that influence and knowledge are two different concepts. Influence is local and it is only affected by direct citations, while knowledge is global and any change in the citation network may cast impact on it. Therefore, influence cannot reflect the importance of a paper's location in the citation network, that is, different from knowledge. Knowledge can be thought of as a structural derivation of influence. In that sense, quantifying knowledge, not just influence, can help understand the knowledge value of articles, reveal the development of academic knowledge brought about by the expansion of the discipline, and further alleviate this contradiction between knowledge acquisition and information generation. 

Although the measurement of knowledge is important, however, the definition of knowledge in computer science, especially the quantification of knowledge, remains largely unexplored. Since the Gettier problem \cite{Gettier} was proposed, scientists have been arguing on the definition of knowledge for more than half a century. In 1989, Ackoff talks about the DIKW pyramid \cite{ackoff1989data} (data, information, knowledge, and wisdom) and describes the positioning of knowledge, which is very instructive for us. Knowledge is based on information, and the quantification of information has quite mature theories, such as Shannon's entropy\cite{shannonentropy}, Li Angsheng's structural entropy\cite{structuralentropy} and so on. Therefore, information theory can serve as a clue to the quantification of knowledge.

\section{Our Contributions}

In this article, we leverage the graph to serves as the best structural representation of abundant underlying correlations in massive data and aim to explore the amount of knowledge contained in the whole academic citation network, where citations between scientific articles imply a knowledge flow from the authors of the article being cited to the authors of the articles that cite it.

Motivated partly by the structural entropy \cite{structuralentropy}, whose idea is to find a partitioning tree that can best represent the information contained in the graph with minimum uncertainty of the graph structure involved. Analogously, we propose the Quantitative Index of Knowledge (KQI), taken as the subtraction between structural entropy in structured high-dimensional space and Shannon entropy, its specialized version in unstructured probability space, to reflect the extent of disorder difference (amount of knowledge) caused by structure (order). 

We explicitly derive KQI by virtue of our proposed Knowledge Tree, a partitioning tree constructed by preserving important knowledge inheritance relationships from the original graph and disclose some interesting phenomenon: \textbf{(1)} the polynomial growth of the graph size overall leads to a linear knowledge growth over time, but meanwhile \textbf{(2)} there exists a threshold of the mean degree $m \sim a\log n + 1$, above which, the knowledge growth accelerates significantly ($a$ is the required number of active neighbors for each node to be active and $n$ is graph size). \textbf{(3)} We also find a similar 80/20 rule in academic citation networks, where a small amount of literature represents the majority of knowledge. These phenomenons not only reveal the development rule of knowledge but also tell the significance of ordinary literature, whose tremendous proliferation eventually brings continuous progress on knowledge.

Experiments of academic citation networks, which cover hundreds of millions of literature in 19 disciplines published during the last 50 years, not only confirm our theoretical analysis but also reveals the effectiveness of KQI in characterizing the knowledge development pattern in different disciplines and finding out important literature, authors, affiliations, and countries of each discipline. Compared with citations, h-index, impact factors, and PageRank, KQI all performs obvious advantages:

\begin{itemize}
    \item Compared with citations, KQI retrieves the valuable articles with few citations and filter the useless articles with many citations.
    \item Compared with h-index, KQI better covers the high-impact authors, like Nobel Prize and Turing Award winners.
    \item Compared with impact factor, KQI evaluates articles more justly. Articles in a good journal can be bad, and articles in a bad journal can be good.
    \item Compared with PageRank, KQI keeps many of the same advantages. But KQI is better on interpretability, formulation, complexity, and additivity.
\end{itemize}

In the remaining chapters of this article, we will elaborate on the above contributions in turn.

Chapter \ref{Chapter:KQI} starts from the question of what is knowledge, and gradually introduces the Quantitative Index of Knowledge (KQI), which is defined as the difference between Shannon entropy and structural entropy. 

Chapter \ref{Chapter:KQIKT} puts forward the concept of Knowledge Tree, and then explicitly deduces the formula of KQI in the knowledge tree. Detailed derivations, proofs, and algorithms are presented in this chapter.

Chapter \ref{Chapter:rule} discusses the linear growth law, knowledge boom threshold, and 80/20 rule of KQI, through many experiments with a dataset of more than 185 million literature.

Chapter \ref{Chapter:application} introduces some applications of KQI in structure extraction and ranking, and compares them with other related metrics. The results show that KQI has obvious advantages in these tasks.

Chapter \ref{chapter:casestudy} takes the field of channel capacity, deep learning, and geoscience as examples, and demonstrates some specific detailed results and analysis.

The last chapter summarizes the contributions and future directions of this article.

% !TEX root = ../main.tex

\chapter{Quantitative Index of Knowledge}
\label{Chapter:KQI}
In this chapter, we will introduce our newly proposed concept, Quantitative Index of Knowledge (KQI). In doing so, based on the JTB theory and the characteristics of the academic network, we first give the definition of knowledge in a graph network. Then the concept of structural entropy is introduced and the intuition that KQI is measured by the difference between Shannon entropy and structural entropy is put forward and explained.

\section{What is Knowledge}
Plato once proposed the JTB theory (the view that knowledge is justified true belief) thousands of years ago\cite{JTB}, although Gettier questioned it\cite{Gettier}, and since then there is no accurate definition of knowledge. Numerous studies have shown the structure of knowledge\cite{knowledgestructure1,knowledgestructure2}, and the important role of the network in explaining knowledge\cite{knowledgenetwork}. Considering a large number of association relationships in academic data, such as citations, we modeled these in an academic citation network\cite{citation1} and tried to find knowledge in this structured space. Academic citation network can be seen as an evolving directed acyclic graph, in which each node represents knowledge to be measured and each edge represents the succession relationship of knowledge. The meaning of the word evolving is that a new node, as well as its inheritance relationship, can be added into the network, and the weights of old node relations decay over time. Coincidentally, some indications of knowledge are seen in the citation network (Fig. \ref{fig:citationnet}): 
\begin{itemize}
    \item The subsequent citations of the paper reflect the widespread recognition of the paper, i.e. relatively \textbf{truth}\cite{relativetruth} in the network, also similar to the "the relativity of knowledge"\cite{knowledgerelativity1,knowledgerelativity2} proposed in the epistemology.
    \item The references of literature also reflect whether the source of the paper is reliable and \textbf{justified}.
\end{itemize}

\begin{figure}[!htp]
  \centering
  \includegraphics[width=10cm]{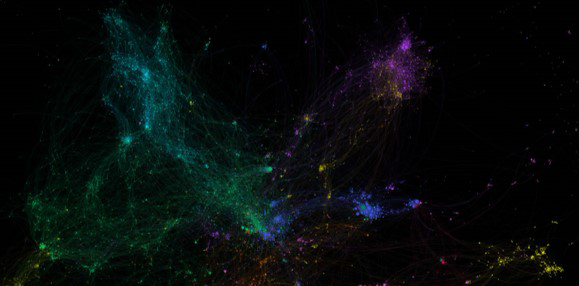} \\
  \caption{Academic citation network.}
 \label{fig:citationnet}
\end{figure}

Therefore, \textbf{knowledge can be expressed as the paper and the structure on which it depends}, though it may be hard to understand. In other words, the effect brought by knowledge is embodied in the association of academic networks. We cannot tell knowledge that is not belonging to the structure, just like we cannot tell a node that is outside of a topological network. Based on this assumption, we present several characteristics of knowledge in the evolutionary network:
\begin{itemize}
    \item \textbf{Knowledge has strict hierarchies and no circular reasoning.} Although this puts forward higher requirements on the reliability of data sources, the data in the knowledge domain is often of higher quality, which is different from the general big data.
    \item \textbf{Knowledge evolves over time.} Knowledge is not fixed, which means that new knowledge is constantly being included, to refresh the existing knowledge structure.
    \item \textbf{Knowledge can be forgotten.\cite{ageing1,ageing2}} The aging of knowledge is manifested as the decline of the inheritance relation between new knowledge and old knowledge. That is it cannot be recalled from present knowledge of its origin.
    \item \textbf{Structure matters than knowledge itself.} The value of knowledge depends more on its position in the knowledge structure than on its content. And the structure is more objective than the content.
\end{itemize}

Now that we have a basic understanding of the knowledge in the graph, then we can talk about how do we measure knowledge.

\section{How to Measure Knowledge}
\label{sec:KQI}

To better understand the following argument, let us first understand the relationship between information, entropy, and knowledge.

The amount of information in the world is unknown, but the amount of information in the human eye is limited. The amount of information under human vision is expanded through discovering the unknown, and the amount of knowledge increases through regularizing information. Discovering and regularizing are the two important processes for human beings to accept new things and form new knowledge. There are three states in these two processes: unknown information, unstructured discovered information, knowledge. Among these, we have known that unstructured discovered information plus knowledge correspond to Shannon entropy which defines the uncertainty of information (Fig. \ref{fig:intuition}). Therefore if we want to measure knowledge, we just need to quantify unstructured discovered information. Coincidentally, in 2016, Prof. Ansheng Li proposed structural entropy\cite{structuralentropy}, whose core is to quantify structural information. Unstructured discovered information corresponds to structure entropy which defines the uncertainty of structured information,

\begin{figure}[!htp]
  \centering
  \includegraphics[width=7cm]{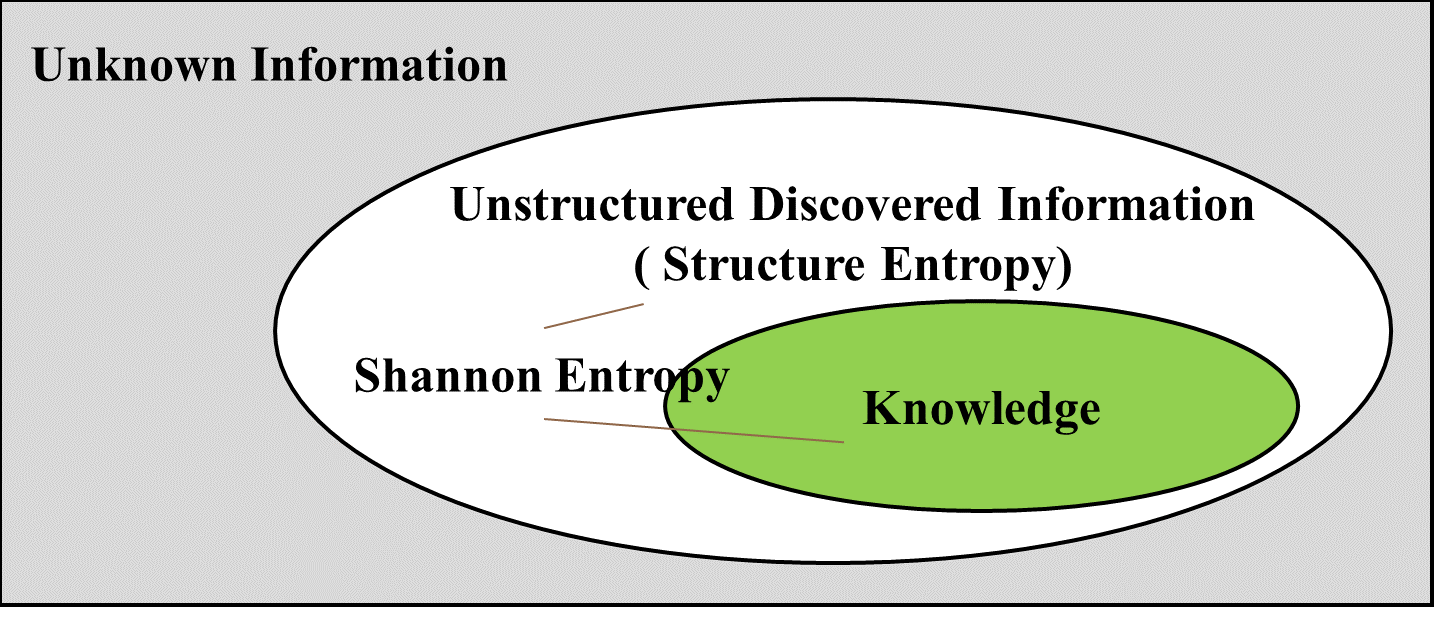} \\
  \caption{Venn diagrams of information, entropy, and knowledge.}
 \label{fig:intuition}
\end{figure}

In this section, in order to better understand the rest of this article, we will first briefly introduce the related concepts of structural entropy, and then explain how knowledge is measured in detail.

\subsection{Structural Entropy}

As mentioned above, structure entropy\cite{structuralentropy} is proposed by Angsheng Li in 2016 to measure the structural information and complexity of the network. Different from Shannon entropy, which only considers probability distribution, structure entropy takes into account the layered community structure in the network. Its core is to use the structure as the inherent information to compress the length of the encoding.

\begin{definition}[Volume\cite{structuralentropy}]
The volume of a community $\alpha$ is the sum of the weighted degree of all the nodes within it, defined as:
$$V_{\alpha} = \sum_{i\in\alpha}d_i,$$
where $d_i$ is the weighted degree of node $i$. 
\end{definition}

\begin{definition}[One-Dimensional Structural Information \cite{structuralentropy}]
The one-dimensional structural information of connected and undirected graphs is defined as follows:

\begin{equation}
    \mathcal{H}^1(G) = H\left(\frac{d_1}{2m},...,\frac{d_n}{2m}\right) = -\sum_{i=1}^{n}\frac{d_i}{2m}\log_2 \frac{d_i}{2m},
\end{equation}

where $n$ and $m$ represent the number of nodes and edges of $G$ respectively, $d_i$ is the degree of the  $i$-th node. This definition is essentially the Shannon entropy of the degree distribution in the graph $G$.
\end{definition}

\begin{definition}[Two-Dimensional Structural Information \cite{structuralentropy}]
Given a connected and undirected graph $G(V, E)$ and a partition $\mathcal{P} = \{X_1,...,X_L\}$ of $V$,  the two-dimensional structural information is defined as follows:
$$\mathcal{H}^{\mathcal{P}}(G) = -\sum_{j=1}^L \frac{V_j}{2m}\sum_{i=1}^{n_j}\frac{d_i^{(j)}}{V_j}\log_2\frac{d_i^{(j)}}{V_j} - \sum_{j=1}^L \frac{g_j}{2m}\log_2\frac{V_j}{2m},$$
where $L$ is the number of modules in partition $\mathcal{P}$, $n_j$ is the number of nodes in module $X_j$, $d_i^{(j)}$ is the degree of the $i$-th node of $X_j$, $V_j$ is the volume of module $X_j$, and $g_j$ is the number of edges with one endpoint in module $X_j$. This formula is equivalent to the code of each module plus the code of each node within the communities. That is, all nodes within a module share a common community code.
\end{definition}

\begin{definition}[High-Dimensional Structural Information \cite{structuralentropy}]
Suppose that $\mathcal{T}$ is a partitioning tree of a connected and undirected graph $G(V, E)$, then the high-dimensional structural information is defined as follows:

\begin{equation}
    \mathcal{H}^{\mathcal{T}}(G) = \sum_{\alpha\in\mathcal{T}, \alpha\ne\lambda} - \frac{g_\alpha}{2m}\log_2\frac{V_\alpha}{V_{\alpha^-}},
\end{equation}

where $g_\alpha$ is the number of edges from nodes in $T_\alpha$ to nodes outside $T_\alpha$, $V_\alpha$ is the volume of set $T_\alpha$. This definition is an extension of two-dimensional structural information and reuses the code of common nodes to a greater extent.
\end{definition}

\subsection{Intuitions for Entropy Difference}

Intrinsically, structural entropy measures the extent of disorder of a structured network. Thus the idea from Fig. \ref{fig:intuition} is that we can differentiate between the Shannon entropy and structural entropy by knowledge. In fact, it is easy to understand that the function of knowledge is to organize disordered data into ordered data.\cite{knowledgefunction1,knowledgefunction2} In physics, entropy is to measure the extent of disorder\cite{entropy}, so entropy can also be used to measure knowledge. Shannon entropy\cite{shannonentropy}, or one-dimensional structural information $H^1$, measures the extent of disorder in the discrete probability distribution, while structural entropy\cite{structuralentropy}, or high-dimensional structural information $H^T$, measures the extent of disorder after organizing the discrete one into a structured network. The above two exactly correspond to the process of knowledge turning disorder into order. 

More fundamentally, from the perspective of coding, structural entropy realizes the sharing of community encodings because of the existence of a partitioning tree, i.e. community inheritance relation. This sharing mechanism saves the number of encodings, which is the value that structure brings (Fig. \ref{fig:kqi}). Therefore the difference between the two entropies is the role of knowledge in it, which we define as the Quantitative Index of Knowledge (KQI).
\begin{equation}
    \mathcal{K}^{\mathcal{T}} = \mathcal{H}^1 - \mathcal{H}^{\mathcal{T}}.
\end{equation}

\begin{figure}[!htp]
  \centering
  \includegraphics[width=7cm]{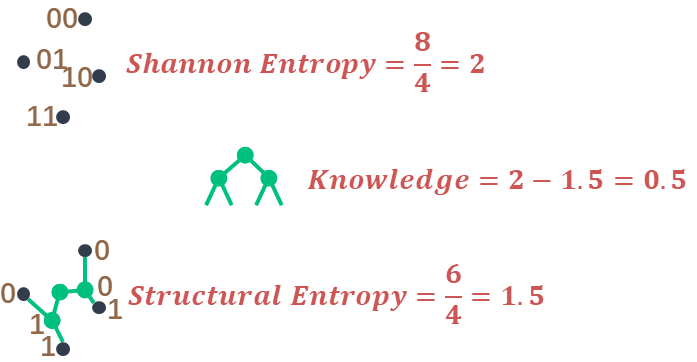} \\
  \caption[Three information-related quantitative indicators.]{Three information-related quantitative indicators. Shannon entropy encodes a discrete probability distribution, where it takes an average of 2 bits to uniquely identify an object. Structural entropy takes structure into account as opposed to Shannon entropy, and it takes an average of 1.5 bits because of the shared encoding caused by the structure. The difference between these two entropies is precisely the difference caused by structure, namely the KQI.}
 \label{fig:kqi}
\end{figure}

Based on the above definition, obtaining KQI also depends on an obbligato link: the partitioning tree $\mathcal{T}$, which is hard to get. In actual application, according to the original method of structure entropy, the algorithm to find the optimal partitioning tree has no guarantee of complexity, which cannot be directly applied to the big data scenario. Now the only thing we still need to be clear about is the partitioning tree $\mathcal{T}$.

Given the large number of structural relationships implied in knowledge, it may not be necessary to explore the optimal community structure. We can get KQI directly from the structure of the knowledge, rather than looking for the potential structure of the knowledge network. Then how to find the structure of knowledge explicitly? In the research of structural entropy, the structure of a network is specified using a partitioning tree\cite{structuralentropy}, which is a kind of hierarchical partitioning community. This is very similar to the structure of knowledge: any knowledge is either inferred from some existing knowledge (belongs to the parent knowledge community) or is a pure axiom. However, unlike the partitioning tree with only one root community\cite{structuralentropy}, the structure of knowledge can be seen as a combination of many partitioning trees, because we have a lot of axiomatic knowledge, and these partitioning trees overlap each other because a piece of knowledge can be inspired by multiple knowledge, i.e. belong to multiple knowledge communities simultaneously. We will elaborate on how to obtain a partitioning tree that is appropriate for representing knowledge in the next chapter.

% !TEX root = ../main.tex

\chapter{KQI in Knowledge Tree}
\label{Chapter:KQIKT}
In the previous chapter, we defined KQI as the difference between Shannon's entropy and structural entropy. Due to the diversity and complexity of structural entropy caused by the partitioning tree, we attempt to define an interpretable tree as an example to explicitly give the expression of KQI. In order to retain the original inheritance relationship of knowledge as much as possible, we model this inheritance relationship as a Knowledge Tree, and further implement a kind of KQI by using the knowledge tree as the partitioning tree.

It is worth noting that a knowledge tree is different from a partitioning tree. Each node in the knowledge tree corresponds to both a node and a community (the structure on which the node depends) in the partitioning tree, that is, each node in the knowledge tree corresponds to a piece of knowledge. Therefore, \textbf{knowledge is used indiscriminately in this chapter to describe the nodes in the knowledge tree.} The detailed transformation of the knowledge tree and partitioning tree will be presented in Section \ref{sec:knowledgetreeaspartitioning}.

In this chapter, we first give the definition of the knowledge tree and then give the formula of KQI and its proof. To make it easier for the reader to understand, we will start with the simplest single inherited n-source knowledge tree and then move on to the general knowledge tree. In the end, the algorithm of KQI is given.

\section{Definition of Knowledge Tree}
\label{sec:knowledgetree}
We have previously claimed that academic citation networks can be considered as a directed acyclic graph, in which a knowledge inheritance tree can be obtained from either breadth-first search or depth-first search starting from any source node. Such a knowledge inheritance tree can be used as a partitioning tree, though it can not adequately represent the inheritance relationship. In this simplest case, the knowledge tree is just a tree. Although this situation hardly works in practice, it can be a good transition to understanding the general definition of the knowledge tree. 

Just to make it easier for the reader, considering multiple source nodes, we first give the definition of a single inherited knowledge tree that can be applied.

\begin{definition}[Single-Inherited Knowledge Tree]
For a directed acyclic graph with $n$ source nodes, at most one parent node is selected for each node, and a super root node is introduced to connect to all the $n$ source nodes. Then we obtain the single-inherited knowledge tree.
\end{definition}

A single-inherited knowledge tree is reflected as essentially a forest. We say it is a knowledge tree by assuming that all the source nodes, that is, axiomatic knowledge, originate from something more primitive, which we call the super root node. A single inheritance tree expresses that every knowledge comes from one and only one source. It is important to note that this assumption is different from our common sense. We expect to be able to consider multiple knowledge sources and multiple different partitioning trees simultaneously. In order to make full use of the knowledge structure, we introduce the multiple-inherited knowledge tree. 

\begin{definition}[Multiple-Inherited Knowledge Tree]
For a directed acyclic graph with $n$ source nodes, multiple important parent nodes are selected for each node, and a super root node is introduced to connect to all the $n$ source nodes. Then we obtain the multiple-inherited knowledge tree, or for short, knowledge tree.
\end{definition}

Because we just give the definition straight away above, a multiple-inherited knowledge tree does not look like an actual tree. We will go into more detail about what kind of tree does it represents.

Decomposing a knowledge tree into an ordinary tree is done by diving nodes from leaves to roots. One observation is that loops in the knowledge tree are impossible, which ensures that decomposition is completed in finite steps. Regardless of the super root node, the knowledge tree should be a forest composed of $n$ trees, but there is some overlap between the $n$ trees that we do not want. The overlap between the trees comes from the crossover, that is, the in-degree of a node is greater than one. Crossover on node $v$ indicates that the knowledge sub-tree starting from $v$ belongs to the common result of multiple parent knowledge (nodes), each of which occupies a certain proportion. Thus, decomposing is to divide each node and its sub-tree into parts equal to in-degree, and we will get $n$ ordinary trees (Fig. \ref{fig:knowledgetree}). In the same way as a single-inherited knowledge tree, when we introduce the root node, the knowledge tree will look like an actual tree.

\begin{figure}[!htp]
  \centering
  \includegraphics[width=12cm]{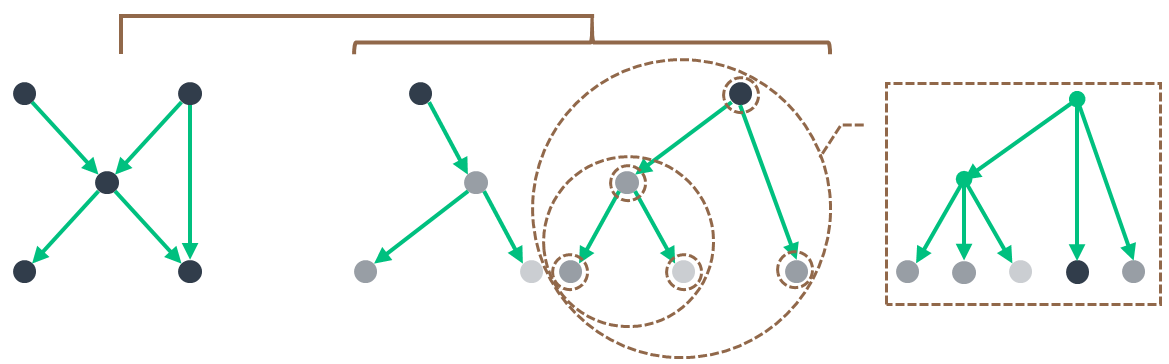} \\
  \caption[Knowledge tree, decomposing process, and transformation to partitioning tree.]{Knowledge tree, decomposing process, and transformation to partitioning tree. Multiple trees can be split from a knowledge tree, with volume assigned (color depth of the node). The nodes that broken up eventually become fragments with less weight. Trees imply the structure of layered communities (brown dotted circle), i.e. the partitioning tree (tree in the brown dotted box). The green nodes in partitioning tree imply the community rather than actual nodes.}
 \label{fig:knowledgetree}
\end{figure}

Here we can find that there seem to create more nodes in the knowledge tree. To explain this, we need to introduce a concept in the multiple-inherited knowledge tree: fragment.
\begin{definition}[Fragments]
Fragments are defined as new nodes created by splitting the node with multiple sources. One Fragment is a part of the node and belongs to only one source at the same time.
\end{definition}

Therefore, a multiple-inherited knowledge tree can consider all inheritance relationships of a graph by fragmenting the knowledge. Obviously, there is a trade-off between accuracy and efficiency, when we taking into account several important pieces of knowledge. In particular, an untreated directed acyclic graph can itself be considered as a knowledge tree, which completely retains all structural information in the graph with all the parent knowledge (nodes) matters. We can adjust the knowledge tree selection according to the actual requirements of different tasks.

\section{Expression of KQI in Knowledge Tree}

In Section \ref{sec:KQI}, we have recognized the calculation method of KQI from the macroscopic view. Next, we will use the proposed knowledge tree as a partitioning tree to explicitly derive the formula of KQI. As with the introduction of the knowledge tree, we will discuss in detail the formula of KQI and its rationality by using the simplest single inherited knowledge tree and then extend the conclusion to the general knowledge tree.

\subsection{KQI in Single-Inherited Knowledge Tree}
\label{sec:knowledgetreeaspartitioning}
Because a single-inherited knowledge tree is simply an ordinary tree and the knowledge in it will not be divided into multiple fragments, our discussion of formulas will be greatly simplified.

We know that the partitioning tree is a hierarchy of nodes, a tree of virtual existence. Each non-leaf node in the partitioning tree represents a community of nodes, and only the leaf nodes actually exist. Although each node in the knowledge tree corresponds to one of real existence, we say that the knowledge tree can be used as a partitioning tree because the knowledge tree naturally represents the structure of knowledge. Ancestral knowledge each forms its own knowledge community, much like the hierarchical community of the partitioning tree. The partitioning tree represented by the knowledge tree looks like this (Fig. \ref{fig:knowledgetree}):
\begin{itemize}
    \item All nodes in the knowledge tree are leaf nodes in the partitioning tree.
    \item All non-leaf nodes in the knowledge tree are duplicated as non-leaf nodes in the partitioning tree.
    \item Each node $v$ has $d_v^{out}+1$ sub-communities, that is $d_v^{out}$ derived knowledge communities plus knowledge $v$ itself.
\end{itemize}

Using the knowledge tree as a partitioning tree in this way, we give the formula for KQI.

\begin{theorem}[KQI in Single-Inherited Knowledge Tree]
\label{theorem:kqi_single}
Given single-inherited knowledge tree $\mathcal{T}$, the KQI expression for the node representing community $\alpha$ is as follows:
\begin{equation}
\label{eq:KQI_KT}
    \mathcal{K}_\alpha^{\mathcal{T}} = -\frac{V_\alpha}{W}\log_2 \frac{V_{\alpha}}{V_{\alpha^-}},
\end{equation}
where $V_\alpha$ is the volume of community $\alpha$, $W$ is the graph size and $\alpha^-$ is the parent community of $\alpha$.
\end{theorem}

\subsubsection{Proof of Theorem \ref{theorem:kqi_single}}

\begin{lemma}
\label{lemma:equivalenttree}
Any partitioning tree $\mathcal{T}$ of height $h$ can be reconstituted as $\mathcal{T}'$ to satisfy the following conditions: 1) For any height $1 \le h_i \le h$, the communities $\alpha_1,...,\alpha_{n_i}$ represented by all $n_i$ nodes of height $h_i$ are mutually exclusive, and $\bigcup_{1 \le j \le n_i}\alpha_j = S$, where $S$ is the set of all leaf nodes. 2) Before and after the reconstitution, the entropy of the original nodes and the overall entropy of the partitioning tree remains unchanged. We say that $\mathcal{T}$ and $\mathcal{T}'$ are equivalent.
\end{lemma}
\begin{proof}[Lemma \ref{lemma:equivalenttree}]
The proof of this lemma is very simple. Since the partition tree is a tree, it must satisfy the mutual exclusion. And the layer of height 1 must satisfy $\bigcup_{1 \le j \le n_1}\alpha_j = S$. Then, if the layer of height 2 has $\bigcup_{1 \le j \le n_2}\alpha_k = S'$ and $S' \ne S$, create a copy $\alpha'$ in the second layer for each $\alpha \in S-S'$ in the first layer. In this process, the entropy of $\alpha$ is transferred to $\alpha'$ and the entropy of $\alpha$ becomes 0, where the whole thing stays the same. Then go on to the layer of height 3 and so on until all the layers are satisfied.
\end{proof}

\begin{proof}[Theorem \ref{theorem:kqi_single}]
Suppose the knowledge tree $\mathcal{T}$ of height $h$ is selected as the partitioning tree. For one node $\alpha\in\mathcal{T}$, i.e. a community of $V$, denote $V_\alpha$ as volume of community $\alpha$, and $g_\alpha$ as boundary of community $\alpha$. More rigorously,
$$V_\alpha = \sum_{v\in \alpha} d_v^{out}, \quad g_\alpha = \sum_{(i,j)\in E, i\notin\alpha, j\in\alpha} w_{ij}.$$
Denote the parent of $\alpha$ in $\mathcal{T}$ as $\alpha^-$, the parent of $\alpha^-$ as $\alpha^{(-2)}$, and so on. The structure entropy in KEG can be written as:
$$\mathcal{H}^{\mathcal{T}} = \sum_{\alpha\in\mathcal{T}} -\frac{g_\alpha}{W}\log_2 \frac{V_{\alpha}}{V_{\alpha^-}},$$
where $W = \sum_{(i, j)\in E} w_{ij}$.

According to Lemma \ref{lemma:equivalenttree}, we can divide the nodes in $\mathcal{T}$ into $h$ layers, denoted as $l_i (1 \le i \le h)$, and we have $\forall 1 \le i \le h, \bigcup_{\alpha\in l_i} \alpha= V$. Thus, we can rewrite the expression for structure entropy as follows:
$$\mathcal{H}^{\mathcal{T}} = \sum_{1 \le i \le h}\sum_{\alpha\in l_i} -\frac{g_\alpha}{W}\log_2 \frac{V_{\alpha}}{V_{\alpha^-}}.$$

As boundary of community $\alpha$, $g_\alpha$ can also be rewritten as: $g_\alpha = \sum_{v \in \alpha} d_v^{in} - V_\alpha$. Then the formula becomes two parts:
\begin{equation*}
\begin{split}
    First Term &= \sum_{1 \le i \le h}\sum_{\alpha\in l_i} -\frac{\sum_{v \in \alpha} d_v^{in}}{W}\log_2 \frac{V_{\alpha}}{V_{\alpha^-}}\\
    &= \sum_{1 \le i \le h}\sum_{\alpha\in l_i}\sum_{v \in \alpha} -\frac{d_v^{in}}{W}\log_2 \frac{V_{v^{(-i+1)}}}{V_{v^{(-i)}}}\\
    &= \sum_{v \in V} -\frac{d_v^{in}}{W} \sum_{1 \le i \le h}\log_2 \frac{V_{v^{(-i+1)}}}{V_{v^{(-i)}}}\\
    &= \sum_{v \in V} -\frac{d_v^{in}}{W}\log_2 \frac{V_{v}}{V_{v^{(-h)}}}
\end{split}
\end{equation*}
where $V_v = d_v^{out}, V_{v^{(-h)}} = W$. The first term is exactly Shannon entropy of graph. So, we have
$$\mathcal{K}^{\mathcal{T}} = SecondTerm = \sum_{1 \le i \le h}\sum_{\alpha\in l_i} -\frac{V_\alpha}{W}\log_2 \frac{V_{\alpha}}{V_{\alpha^-}}.$$

Applying above lemma again, we give the formula of KQI as follows:
$$\mathcal{K}^{\mathcal{T}} = \sum_{\alpha\in\mathcal{T}} -\frac{V_\alpha}{W}\log_2 \frac{V_{\alpha}}{V_{\alpha^-}}.$$

The proof is completed.
\end{proof}

\subsubsection{Interpretation of KQI formula}
\label{sec:explainKQI}
Next, let us understand what KQI means from the Eq. \eqref{eq:KQI_KT}. 

At the micro-level, in terms of encodings, each knowledge in the knowledge tree, like a structural pivot, plays a role in encoding reuse. KQI is a measure of how many encodings of each knowledge can be reused. The degree of reuse is expressed as the encoding length ($-\log_2 \frac{V_{\alpha}}{V_{\alpha^-}}$) times the probability of reuse ($\frac{V_\alpha}{W}$).

At the macro-level, the structure is a form of organization, such as a hierarchy and a network. A knowledge tree, equivalent to multilevel classification, is a hierarchical structure based on knowledge evolution. The difference between Shannon entropy and structure entropy strips away discrete information other than the structure itself. When we use the knowledge tree as the partitioning tree, KQI only retains the structure of the knowledge tree. In this sense, KQI can be used as an expression of how strongly structured the knowledge is.

Starting with the concept of knowledge, KQI is related to acceptability ($\frac{V_\alpha}{W}$) and dependability ($-\log_2 \frac{V_{\alpha}}{V_{\alpha^-}}$). Acceptability refers to whether knowledge is recognized, i.e. how much knowledge is inherited directly or indirectly from that knowledge. Dependability refers to whether the source of knowledge is equally or more recognized, i.e. how fully the parents can support the generation of the knowledge. Acceptability and dependability, elements of scientific knowledge\cite{scientific1,scientific2}, correspond to the first and second terms of the Eq. \eqref{eq:KQI_KT}, as well as to truth and justification as we talked about earlier in JTB theory\cite{JTB} (Fig. \ref{fig:acceptibility_dependability}).

\begin{figure}[!htp]
  \centering
  \includegraphics[width=7cm]{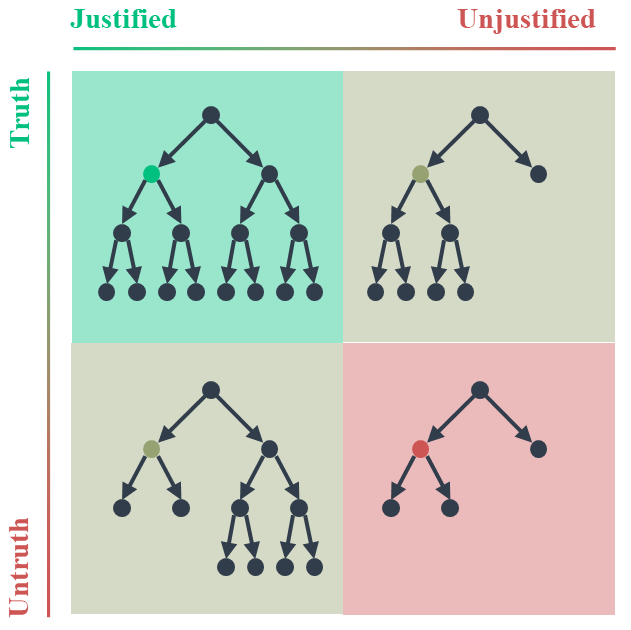} \\
  \caption[KQI-JTB matrix.]{KQI-JTB matrix. For any knowledge, volume means truth, and the difference from parents tells whether it is justified. Knowledge of maximum KQI (green) should be justified truth. Knowledge with at least one truth and justification come next, i.e. unjustified truth or justified untruth. Unjustified untruth has little knowledge (red).}
 \label{fig:acceptibility_dependability}
\end{figure}

We can find the higher the probability of the random walk reaching the node, or the higher the level in the partitioning tree, the more repetitive encoding is reduced and the more strongly structured the node is. In other words, the more connectivity and the more subsequent knowledge, the higher KQI. Intuitively, under the definition of KQI, knowledge with higher KQI is developed from very solid knowledge and can form its own system. If knowledge is derived from the knowledge that is not recognized, its KQI is small and can be understood as unacknowledged.

\subsection{Generalized KQI in Knowledge Tree}

Now we have the formula for KQI in the simplest single-inherited knowledge tree. In this part, we will extend KQI to a general knowledge tree.

In the discussion of Section \ref{sec:knowledgetree}, we know that a knowledge tree can also be regarded as a real tree. But the process of decomposing into a tree can lead to some splitting of knowledge. So if we want to get the KQI of one knowledge, we have to consider all the fragments that are formed by splitting that knowledge. Here we give the definition of KQI in the general knowledge tree.
\begin{definition}[KQI in Knowledge Tree]
\label{def:KQI_KT}
Given knowledge tree $\mathcal{T}$ and denoted $S_\alpha$ as the set of all fragments of $\alpha$, the KQI of $\alpha$ is defined as:
    \begin{equation}
    \label{eq:KQI_fragment}
    \mathcal{K}_\alpha^{\mathcal{T}} = \sum_{\alpha_i\in S_\alpha} -\frac{V_{\alpha_i}}{W} \log_2\frac{V_{\alpha_i}}{V_{\alpha_i^-}}.
    \end{equation}
\end{definition}

If we use Eq. \eqref{eq:KQI_KT} directly, this is a fairly complex process because the number of fragments is exponentially related to the depth of knowledge. If each knowledge originated from two parents, the $k$-th generation of knowledge will be split into $2^k$ fragments. Here we present a method equivalent to the above calculation, but with much less complexity.

\begin{theorem}[More Efficient Equivalent KQI in Knowledge Tree]
\label{theorem:KQI}
Given knowledge tree $\mathcal{T}$, the KQI expression defined in Definition \ref{def:KQI_KT} is equivalent to following expression:
\begin{equation}
\label{eq:KQI}
\mathcal{K}_\alpha^{\mathcal{T}} = - \sum_{1 \le i \le d_\alpha^{in}} \frac{V_\alpha^{\mathcal{T}}}{d_\alpha^{in}W} \log_2\frac{V_\alpha^{\mathcal{T}}}{d_\alpha^{in}V_{\alpha_i^-}^{\mathcal{T}}},
\end{equation}
in which
$$V_\alpha^{\mathcal{T}} = d_\alpha^{out} + \sum_{1 \le i \le d_\alpha^{out}} \frac{V_{\alpha_i^+}^{\mathcal{T}}}{d_{\alpha_i^+}^{in}},$$
where $\alpha_i^-$ represents the $i$-th parent of $\alpha$, $\alpha_i^+$ represents the $i$-th child of $\alpha$, $V_\alpha^{\mathcal{T}}$ represents the sum of all the $V_\alpha$ of different fragments in $\mathcal{T}$. It can be proved that Eq. \eqref{eq:KQI} is equivalent to Eq. \eqref{eq:KQI_fragment}.
\end{theorem}

Before proving that the above formulas are equivalent, let us introduce two lemmas first.

\begin{lemma}
\label{lemma:split}
Assume $\beta$ has $n$ fragments $\beta_1,...,\beta_n$, and $\alpha$ is a child of $\beta$, then we have $\frac{V_{\alpha_1}}{V_{\beta_1}}=...=\frac{V_{\alpha_n}}{V_{\beta_n}}$.
\end{lemma}
\begin{proof}[Lemma \ref{lemma:split}]
This lemma makes sense because, in our definition, all of the successive nodes of $\beta$ will be split the same way when we split $\beta$.
\end{proof}

\begin{lemma}
\label{lemma:volume}
Given n-source knowledge tree $\mathcal{T}$, in which knowledge $\alpha$ has $n$ fragments $\alpha_1,...,\alpha_n$, we have $V_\alpha^{\mathcal{T}} = \sum_{i=1}^n V_{\alpha_i}$.
\end{lemma}
\begin{proof}[Lemma \ref{lemma:volume}]
We use inductive reasoning to prove it. (1) If knowledge $\alpha$ is a leaf node, $\alpha$ has no children and no fragments. Obviously, $V_\alpha^{\mathcal{T}} = d_\alpha^{out} = V_\alpha$. (2) Assuming that all children of $\alpha$ satisfy the lemma, next we explain that $\alpha$ also satisfies the lemma. Given $\alpha^+$ that is a child of $\alpha$, we say that each fragment $\alpha_i^+$ is added to $\alpha_i$ in the same proportion $\frac{1}{d_{\alpha_i^+}^{in}}$. So sum over all the $\alpha_i^+$ is the same as $\sum_{1 \le i \le d_\alpha^{out}} \frac{V_{\alpha_i^+}^{\mathcal{T}}}{d_{\alpha_i^+}^{in}}$. And the sum of the out-degree of all the fragments of $\alpha$ is also equal to $d_\alpha^{out}$. Thus, the lemma is proved.
\end{proof}

Now let us prove Theorem \ref{theorem:KQI}.

\begin{proof}[Theorem \ref{theorem:KQI}]
Divide all the fragments of $\alpha$ into $d_\alpha^{in}$ groups, denoted as $S_{\alpha_i^-} (1 \le i \le d_\alpha^{in})$, according to $\alpha$'s parent. The fragment in group $S_{\alpha_i^-}$ is denoted as $\alpha_{(i,j)} (1 \le j \le |S_{\alpha_i^-}|)$. Applying lemma 1, in the same group $S_{\alpha_i^-}$, the proportion $\frac{V_{\alpha_{(i,j)}}}{V_{\alpha_{(i,j)}^-}}$ stays the same. Applying lemma 2, $\frac{V_{\alpha_{(i,j)}}}{V_{\alpha_{(i,j)}^-}} = \frac{\sum_{j=1}^{|S_{\alpha_i^-}|} V_{\alpha_{(i,j)}}}{\sum_{j=1}^{|S_{\alpha_i^-}|} V_{\alpha_{(i,j)}^-}} = \frac{V_{\alpha_i}^{\mathcal{T}}}{V_{\alpha_i^-}^{\mathcal{T}}}$. So we can sum over $\mathcal{K}_{\alpha_{(i,j)}}^{\mathcal{T}}$ in each of these groups separately.
\begin{equation*}
\begin{split}
    \sum_{1 \le j \le |S_{\alpha_i^-}|}\mathcal{K}_{\alpha_{(i,j)}}^{\mathcal{T}} &= \sum_{1 \le j \le |S_{\alpha_i^-}|} -\frac{V_{\alpha_{(i,j)}}}{W} \log_2\frac{V_{\alpha_{(i,j)}}}{V_{\alpha_{(i,j)}^-}}\\
    &= \sum_{1 \le j \le |S_{\alpha_i^-}|} -\frac{V_{\alpha_{(i,j)}}}{W} \log_2\frac{V_{\alpha_i}^{\mathcal{T}}}{V_{\alpha_i^-}^{\mathcal{T}}}\\
    &= -\frac{V_{\alpha_i}^{\mathcal{T}}}{W} \log_2\frac{V_{\alpha_i}^{\mathcal{T}}}{V_{\alpha_i^-}^{\mathcal{T}}}\\
    &= -\frac{V_\alpha^{\mathcal{T}}}{d_\alpha^{in}W} \log_2\frac{V_\alpha^{\mathcal{T}}}{d_\alpha^{in}V_{\alpha_i^-}^{\mathcal{T}}}
\end{split}
\end{equation*}

Then when we add up all the groups, we find that the result is the same as Eq. \eqref{eq:KQI_fragment}.
\end{proof}

\section{KQI Algorithm}

In this part, we give the algorithm of KQI, whose input is a directed acyclic graph. The algorithm consists of two stages: preparation and query. In the preparation phase, lines 1 to 3, go through all the knowledge in reversed topological order, calculate their volumes and the weight of the graph. In the query phase, lines 5 to 6, we apply Eq. \eqref{eq:KQI} in the general knowledge tree to calculate the KQI.

Obviously, the complexity of the preparation phase is $O(n)$, and the complexity of the query phase is $O(1)$.

\begin{algorithm}[htb]
\caption{KQI}
% \small
\SetAlgoLined
\DontPrintSemicolon

\SetKwInput{KwInput}{Input}
\SetKwInput{KwOutput}{Output}
\SetKwFunction{FKQI}{KQI}

    \KwInput{Directed Acyclic Graph $(V, E)$}
    \KwOutput{KQI of each node $v$}

    \For{$v \in$ TopologicalSort($V$)}{
        $W$ = $W$ + $d_v^{out}$\;
        $vol[v]$ = $d_v^{out}$ + $\sum_{v \rightarrow u} \frac{vol[u]}{d_{u}^{in}}$\;
    }
    \;
    \SetKwProg{Fn}{Function}{:}{}
    \Fn{\FKQI{$v$}}{
        \KwRet $\sum_{u \rightarrow v} -\frac{vol[v] / d_v^{in}}{W} \log_2 \frac{vol[v] / d_v^{in}}{vol[u]}$\;
    }
    
%   $\sum_{i=1}^{\infty} := 0$ \tcp*{this is a comment}
%   \tcc{Now this is an if...else conditional loop}
%   \If{Condition 1}
%     {
%         Do something    \tcp*{this is another comment}
%         \If{sub-Condition}
%         {Do a lot}
%     }
%     \ElseIf{Condition 2}
%     {
%     	Do Otherwise \;
%         \tcc{Now this is a for loop}
%         \For{sequence}    
%         { 
%         	loop instructions
%         }
%     }
%     \Else
%     {
%     	Do the rest
%     }
    
%     \tcc{Now this is a While loop}
%   \While{Condition}
%   {
%   		Do something\;
%   }
\end{algorithm}

% !TEX root = ../main.tex

\chapter{Findings on KQI}
\label{Chapter:rule}
In this chapter, we retrieved and integrated known academic database including but not limited to Nature, Science, Elsevier and Springer: more than 185 million literatures published between 1970 and 2020, and over 1 billion citations among them. The literature covers 292 domains in 19 disciplines: Economics, Environmental science, Biology, Computer science, Sociology, Physics, Materials science, Chemistry, Geology, Geography, Medicine, Business, Engineering, History, Political science, Art, Philosophy, Mathematic and Psychology. Detailed statistical indicators for the dataset are listed in Table \ref{tab:dataset}.

\begin{table}[!htpb]
  \caption{Statistical properties of dataset.}
  \label{tab:dataset}
  \centering
  \begin{threeparttable}[b]
     \begin{tabular}{lcrr}
        \hline
            \textbf{First-level disciplines} & \textbf{Subdisciplines} & \textbf{Literatures} & \textbf{Citations}\\
        \hline
            Art & 6 & 8,431,187 & 1,181,039\\
            Biology & 32 & 18,272,052 & 250,601,059\\
            Business & 13 & 5,307,361 & 4,746,419\\
            Chemistry & 21 & 17,558,380 & 139,346,592\\
            Computer science & 34 & 13,647,023 & 60,625,629\\
            Economics & 40 & 4,530,937 & 26,995,163\\
            Engineering & 44 & 11,678,544 & 15,157,197\\
            Environmental science & 8 & 3,301,952 & 1,541,758\\
            Geography & 11 & 5,498,603 & 1,790,338\\
            Geology & 18 & 4,248,257 & 29,095,511\\
            History & 7 & 8,018,243 & 2,875,534\\
            Materials science & 7 & 10,100,461 & 33,002,075\\
            Mathematics & 20 & 7,928,338 & 49,616,839\\
            Medicine & 45 & 31,965,814 & 219,969,265\\
            Philosophy & 7 & 4,547,424 & 1,609,383\\
            Physics & 27 & 11,447,568 & 63,977,908\\
            Political science & 3 & 8,106,899 & 5,085,742\\
            Psychology & 14 & 9,561,040 & 66,966,763\\
            Sociology & 13 & 6,688,744 & 11,632,120\\
        \hline
            \textbf{TOTLE: 19} & \textbf{292\tnote{*}} & \textbf{185,804,569\tnote{*}} & \textbf{1,512,707,248\tnote{*}}\\
        \hline
    \end{tabular}
    \begin{tablenotes}
        \item [*] 1. Among 292 subdisciplines, some belong to multiple first-level disciplines at the same time. \\ 2. For one discipline, we only count the citations with both subject and object in that discipline.% or \item [a]
    \end{tablenotes}
  \end{threeparttable}
\end{table}

In such a large academic dataset, through a large number of experiments, we found the following rules and laws: linear growth law, knowledge boom threshold, and 80/20 rule. We will introduce them in turn next.

\section{Linear Growth Law}

Linear growth law, as the name implies, means that for most disciplines, the total KQI of the whole academic citation network increases almost linearly over time, although the number of papers shows a polynomial growth.

In order to systematically draw the universal relationship between scientific productivity and knowledge amount, We divide the academic data into subdisciplines of varying sizes and build snapshots of the citation networks from year to year. The KQI proposed previously is used to measure the knowledge of each network, and more lateral insights based on this are shown below.

\subsection{Phenomenon}

Buckminster Fuller coined the "knowledge doubling curve", which held that human knowledge would increase at an alarming rate.\cite{knowledgedoubling} However, we find that later terminology is more often referred to as an information explosion rather than knowledge explosion, which may be a misunderstanding (Fig. \ref{fig:Fig_2a}). In 2008, Krishna and Abhay hinted at the potential of the web to balance the information explosion with knowledge acquisition.\cite{webinfoknowledge} Our experiments confirm that knowledge normally increases linearly over time compared to the number of literature in most academic disciplines, large or small, between 1970 and 2020. The magnitude of the increase in knowledge is wildly out of step with scientific productivity, and the trend continues unabated over time (Fig. \ref{fig:linear}a, Fig. \ref{fig:linear}b). We further theoretically deduced under the BA model\cite{BAmodel} that the polynomial increase in literature brings a linear increase in knowledge. If $n$ represents the number of literature, $m$ represents the average number of references, $t$ represents time, and $k$, $b$ are parameters, we have 
\begin{equation}
\label{eq:linear}
    \sqrt[n]{m} \sim kt+b.
\end{equation}

\begin{figure}[!htp]
  \centering
  \includegraphics[width=5cm]{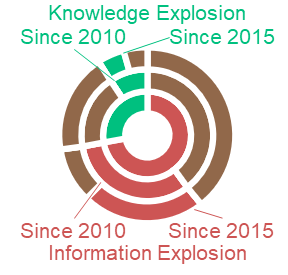} \\
  \caption[Word frequency of \textit{knowledge explosion} and \textit{information explosion}.]{Word frequency statistics of \textit{knowledge explosion} and \textit{information explosion} in Google Scholar. According to statistics in Google Scholar, the word frequency of "information explosion" is more than twice that of "knowledge explosion", and this gap has more than tripled since 2010 and more than quadrupled since 2015.}
 \label{fig:Fig_2a}
\end{figure}

\begin{figure}[!hbtp]
    \centering
    \subcaptionbox{Polynomial literatures.}[4.5cm]{
        \includegraphics[height=4cm]{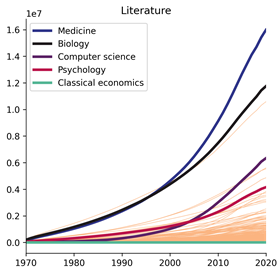}
    }
    \hspace{0.2cm}
    \subcaptionbox{Linear knowledge.}[4.5cm]{
        \includegraphics[height=4cm]{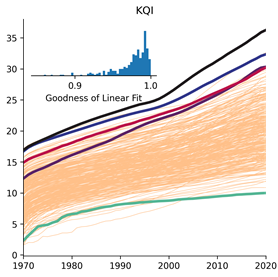}
    }
    \hspace{0.2cm}
    \subcaptionbox{KQI under attenuation.}[4.5cm]{
        \includegraphics[height=4cm]{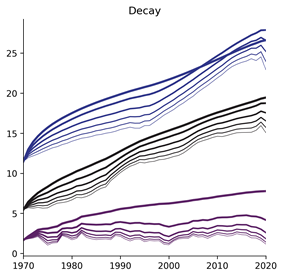}
    }
    \caption[Polynomial literature brings linear knowledge.]{Polynomial literature brings linear knowledge. \textbf{a-b,} Although the number of literature in various 311 disciplines show polynomial growths, the linear goodness of fit of KQI in 86.8\% of them reached 0.95. The five disciplines marked in bold reflect that: the number of literature is related to KQI, but not absolutely. Medicine has the highest number of literature, but biology has the highest KQI, followed by computer science and psychology, and classical economics is the lowest. \textbf{c,} Assign each edge in the citation network with weight $e^{-\lambda(t-t_0)}$, where $t$ is the current time and $t_0$ is the publication time. The lines from thick to thin represent the cases where $\lambda$ is 0, 0.2, 0.4, 0.6, 0.8, and 1, respectively. As the degree of aging increases, the growth becomes shaky but remains linear.}
    \label{fig:linear}
\end{figure}

Many studies about citations take the effect of decay into consideration, with the aging of academic literature as a consensus gradually.\cite{ageing1,ageing2} By using one general attenuation strategies $w(t)=e^{-\lambda(t-t_0)}w_0$, we find that when setting different attenuation coefficients $\lambda$, the same result could be observed. As the attenuation coefficient increases, the growth line becomes shaky, because the recent changes in the graph are more obvious when the accumulation of the earlier years is eliminated (Fig. \ref{fig:linear}c).

Different academic network structures in different disciplines also lead to obvious differences in knowledge amount. The knowledge rankings of the disciplines do not correspond to their scientific productivity rankings (Fig. \ref{fig:Fig_2l}). The results in Fig. \ref{fig:linear}a and Fig. \ref{fig:linear}b show that the best-developed discipline in the last 50 years is biology, though the discipline with the largest number of papers is medicine. Another point of concern is that computer science has developed rapidly in recent years, and we can see that both the number of papers and KQI gradually surpass many other disciplines. Something counter-intuitive is that psychology that doesn't have a lot of papers develops very well in terms of knowledge.

\begin{figure}[!htp]
  \centering
  \includegraphics[width=10cm]{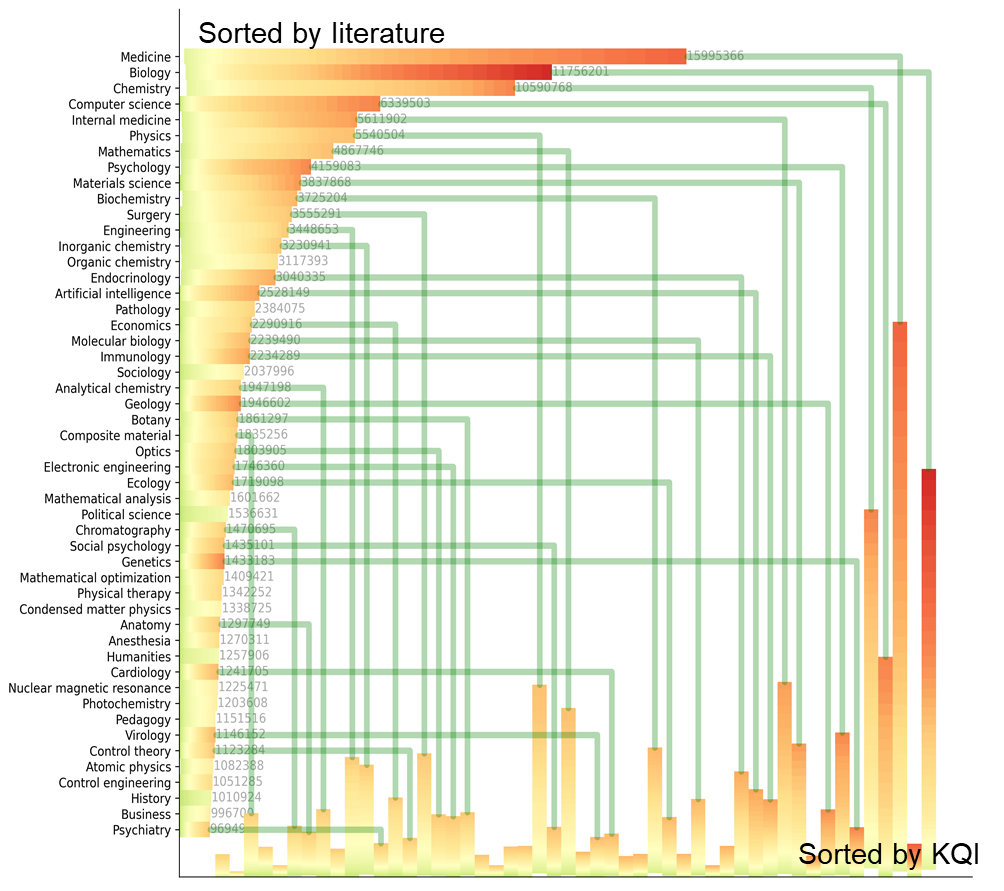} \\
  \caption[Ranking comparison between literatures and KQI.]{Ranking comparison between the number of literatures and KQI. Sorting each discipline according to the number of literatures and KQI, there is a certain correlation between them.}
 \label{fig:Fig_2l}
\end{figure}

\subsection{Proof}

In this part, we explain the law of such order of magnitude difference between literature and KQI. Considering that the paper citation network is a typical scale-free network, which was mentioned by Derek de Solla Price\cite{citation1} in 1965, we attempt to simulate citation network with some scale-free network models. Coincidentally, the BA model\cite{BAmodel} nicely represents the growth of the citations network, with each newly arrived node connected to some old nodes, just as each new paper will cite some previous papers. Therefore, we will prove it under the assumption of the BA model.

\subsubsection{Assumption and Preparation}

% 前提说明
Assuming that the degree of the network follows the power law distribution, we discuss the characteristics of KQI under the classical BA model. In BA model, the degree of the $i$th node is denoted as $d_i^n$ when the $n+1$th new node is about to arrive, and the probability of the $i$th node to connect with the $n+1$th new node is considered as $p_i^n = \frac{d_i^n}{\sum_{j=1}^n d_j^n}$.\cite{BAmodel}

Considering directed graph and BA model in continuous scenarios, assume that nodes arrive at a speed of $s(t)$, and $s(t)\cdot\mathrm{d}t$ represents the number of nodes arriving at time $t$, with each new node connected to $m$ nodes. The probability density of the node arriving at $t_i$ to connect with the node arriving at time $t$ can be written as
\begin{equation}
    \label{eq: p_i}
    p_i(t) = \frac{md_i(t)}{W(t)},
\end{equation}
where $W(t) = (m+1)\int_0^t s(t_j) \mathrm{d}t_j$. $m$ in the numerator means connecting with other $m$ nodes, $m+1$ comes from the assumption that each node includes a self-ring to ensure the model to work properly. So the initial in-degree of each node is $d_i(t_i) = 1$. An obvious conclusion can be also concluded.
$$\int_0^t \frac{p_i(t)}{m} \mathrm{d}t_i = 1.$$

\subsubsection{Degree $d_i$ of the $i$th Node}
In the period of time $\mathrm{d}t$, the increment of degree can be expressed as
$$\mathrm{d}d_i(t) = p_i(t) \cdot s(t)\mathrm{d}t,$$
where we approximately consider that all $s(t)\mathrm{d}t$ nodes are attached to the node $\alpha_i$ with the same probability $p_i$.

Apply Eq. \eqref{eq: p_i},
\begin{equation*}
\begin{split}
    \frac{\mathrm{d}d_i(t)}{d_i(t)} &= m \frac{s(t)\mathrm{d}t}{W(t)} = \frac{m}{m+1} \frac{\mathrm{d}W(t)}{W(t)}\\
    \ln d_i(t) &= \frac{m}{m+1} \ln W(t) + C\\
    d_i(t) &= C' [W(t)]^{\frac{m}{m+1}}
\end{split}
\end{equation*}

As $d_i(t_i) = 1$, we can get
\begin{equation}
    d_i(t) = \left[\frac{W(t)}{W(t_i)}\right]^{\frac{m}{m+1}} = r_i^{\frac{m}{m+1}}(t)
\end{equation}
where we use $r_i(t)$ to represent $\frac{W(t)}{W(t_i)}$ for convenience.

\subsubsection{Node Containing Proportion $P_{i \prec j}$}
Define the proportion that the $j$th node ultimately belongs to the successor node of the $i$th node as $P_{i \prec j}$. Assuming that the connections between any two nodes are independent of each other, and we have
$$P_{i \prec j} = \int_{t_i}^{t_j} \frac{1}{m} p_k(t_j) P_{i \prec k} \cdot s(t_k)\mathrm{d}t_k.$$
Take the derivative of both sides of this equation with respect to $t_j$,
$$\frac{\mathrm{d}P_{i \prec j}}{\mathrm{d}t_j} = \frac{1}{m} p_j(t_j) P_{i \prec j} \cdot s(t_j) + \int_{t_i}^{t_j} \frac{\mathrm{d}p_k(t_j)}{\mathrm{d}t_j} \frac{1}{m} P_{i \prec k} \cdot s(t_k)\mathrm{d}t_k.$$
On account of $\mathrm{d}p_k(t_j) \approx 0 \cdot \mathrm{d}t_j$ when $t_i$ and $t_j$ are big enough, the second term can be ignored. Thus,
$$\mathrm{d}P_{i \prec j} = \frac{1}{m} p_j(t_j) P_{i \prec j} \cdot s(t_j)\mathrm{d}t_j.$$

Applying Eq. \eqref{eq: p_i},
\begin{equation*}
\begin{split}
    \frac{\mathrm{d}P_{i \prec j}}{P_{i \prec j}} &= \frac{d_j(t_j)}{W(t_j)}  \cdot s(t_j)\mathrm{d}t_j = \frac{\mathrm{d}W(t_j)}{(m+1)W(t_j)},\\
    \ln P_{i \prec j} &= \frac{1}{m+1} \ln W(t_j) + C,\\
    P_{i \prec j} &= C'[W(t_j)]^{\frac{1}{m+1}}.
\end{split}
\end{equation*}

As $P_{i \prec i} = p_i(t_i) = \frac{m}{W(t_i)}$, we can get
\begin{equation*}
    P_{i \prec j} = \frac{m}{W(t_i)}\left[\frac{W(t_j)}{W(t_i)}\right]^{\frac{1}{m+1}} = \frac{m}{W(t_i)}r_i^{\frac{1}{m+1}}(t_j).
\end{equation*}

To ensure the correctness of this approximate formula, the least we can do is make sure that $0 < P_{i \prec j} < 1$. We conclude that this formula works only if $W(t_j) < \frac{[W(t_i)]^{m+2}}{m^{m+1}}$, which is easy enough to satisfy in most cases. 
% Otherwise, let us just assume that the probability is 1. 
Thus,
\begin{equation}
    P_{i \prec j} = \frac{m}{W(t_i)}r_i^{\frac{1}{m+1}}(t_j), \quad when \quad W(t_j) < \frac{[W(t_i)]^{m+2}}{m^{m+1}}.
\end{equation}

\subsubsection{Volume $V_{\alpha_i}$ of Community $\alpha_i$}
If $\alpha_i$ is the community represented by the $i$th node, the volume of $\alpha_i$ is denoted as $V_{\alpha_i}$. When $W(t) \le \frac{[W(t_i)]^{m+2}}{m^{m+1}}$, we have
\begin{equation*}
\begin{split}
    V_{\alpha_i} &= \int_{t_i}^{t} \int_{t_i}^{t_j} p_k(t_j) P_{i \prec k} \cdot s(t_k)\mathrm{d}t_k \cdot s(t_j) \mathrm{d}t_j\\
    &= \int_{t_i}^{t} m P_{i \prec j} \cdot s(t_j) \mathrm{d}t_j\\
    &= \int_{t_i}^{t} \frac{m}{m+1} \frac{m}{W(t_i)}r_i^{\frac{1}{m+1}}(t_j) \mathrm{d}W(t_j)\\
    &= \frac{m^2}{m+2} \left[r_i^{\frac{m+2}{m+1}} - 1\right],\\
\end{split}
\end{equation*}

\begin{equation*}
\begin{split}
    V_{\alpha_i^-} &= \frac{\int_{t_1}^{t_i} V_{\alpha_j} \frac{p_j(t_i)}{m} \cdot s(t_j) \mathrm{d}t_j}{\int_{t_1}^{t_i} \frac{p_j(t_i)}{m} \cdot s(t_j) \mathrm{d}t_j}\\
    &= \frac{\int_{r_1}^{r_i} \frac{m^2}{(m+1)(m+2)}r_i^{\frac{1}{m+1}} \left[r_j^{-\frac{m+2}{m+1}}-1\right]\mathrm{d}r_j}{\int_{r_1}^{r_i} -\frac{1}{m+1}r_i^{\frac{1}{m+1}} r_j^{-\frac{m+2}{m+1}} \mathrm{d}r_j}\\
    &= \frac{m^2}{m+2}\left[\frac{r_1-r_i}{(m+1)(r_i^{-\frac{1}{m+1}}-r_1^{-\frac{1}{m+1}})} - 1\right].\\
\end{split}
\end{equation*}

\subsubsection{Formula Estimation of KQI}

The expected KQI of nodes $\alpha_i$ that arrives at time $t_i$ can be expressed as
$$K_{\alpha_i} = -\frac{V_{\alpha_i}}{W}\log_2 \frac{V_{\alpha_i}}{mV_{\alpha_i^-}},$$
among which
$$\frac{V_{\alpha_i}}{mV_{\alpha_i^-}} = \frac{1}{m} \cdot \frac{r_i^{\frac{m+2}{m+1}} - 1}{\frac{r_1-r_i}{(m+1)(r_i^{-\frac{1}{m+1}}-r_1^{-\frac{1}{m+1}})} - 1}.$$

In practice, $m$ is usually a relatively large value. For the reference network of papers, the average value of $m$ is around 10. And we consider the case where $t_i >> t_1$ and $t >> t_i$. So we make some estimates of the above formulas to reduce the computational complexity.

\begin{equation*}
\begin{split}
    \frac{V_{\alpha_i}}{mV_{\alpha_i^-}} &\approx \frac{r_i^{\frac{m+2}{m+1}} - 1}{\frac{r_1-r_i}{r_i^{-\frac{1}{m+1}}-r_1^{-\frac{1}{m+1}}} - m}
    \approx \frac{r_i^{\frac{m+2}{m+1}} - 1}{r_1r_i^{\frac{1}{m+1}} - m}
    \approx \frac{r_i}{r_1}
    \approx \frac{m}{W(t_i)}.
\end{split}
\end{equation*}

This expression says that for any particular node $\alpha_i$, $\log_2 \frac{V_{\alpha_i}}{mV_{\alpha_i^-}}$ is independent of time $t$. The $\log_2 \frac{V_{\alpha_i}}{mV_{\alpha_i^-}}$ converges gradually with the arrival of new nodes, and can be considered a constant value when $W(t_i)$ is large enough.

Thus, in our theoretical analysis, for convenience, we use an approximate simplified formula for KQI:
$$K_{\alpha_i} = \frac{V_{\alpha_i}}{W}.$$

\subsubsection{Evolution Characteristics of KQI}

Looking at the sum of KQI for all the nodes in the entire network, we have

\begin{equation*}
\begin{split}
    K &= \int_{t_1}^{t} K_{\alpha_i} \cdot s(t_i) \mathrm{d}t_i = \int_{t_1}^{t} \frac{V_{\alpha_i}}{W} \cdot s(t_i) \mathrm{d}t_i\\
    &= \int_{t_1}^{t} \frac{m^2}{(m+1)(m+2)} \left[\frac{W^\frac{1}{m+1}(t)}{W^\frac{m+2}{m+1}(t_i)} - \frac{1}{W(t)}\right] \mathrm{d}W(t_i)\\
    &\approx \int_{t_1}^{t} \left[\frac{W^\frac{1}{m+1}(t)}{W^\frac{m+2}{m+1}(t_i)} - \frac{1}{W(t)}\right] \mathrm{d}W(t_i)\\
    &= \left[-\frac{1}{m+1}\left(1-\frac{W^\frac{1}{m+1}(t)}{W^\frac{1}{m+1}(t_1)}\right) - 1 + \frac{W(t_1)}{W(t)}\right]\\
    &\sim \frac{W^\frac{1}{m+1}(t)}{m}.
\end{split}
\end{equation*}

Then we discuss the relationship between the KQI of the whole network and the growth rate of network nodes based on this. 

First we define the standard growth function of the network.
\begin{definition}[Standard Network Growth] A standard growth network with size $W(t)$, should satisfies the equation
$$\frac{W^\frac{1}{m+1}(t)}{m} = kt+b,$$
where $k, b$ are parameters and $m$ is the average degree of the network.
\end{definition}

This equation is the same as Eq. \ref{eq:linear}. For networks satisfying this equation, the KQI of the whole network increases linearly. The experiments on different disciplines show that the growth law of the number of nodes is basically in line with the standard network growth (Fig. \ref{fig:linear}b).

In addition to standard network growth, there are also accelerated network growth and decelerated network growth. When $t \rightarrow \infty$, if $\frac{\mathrm{d^2}K}{\mathrm{d}t^2} > 0$, then we say the corresponding function $s(t)$ is accelerated growth speed. For example, $s(t) = e^{t}, s(t) = t^{m+2}$, etc. When $t \rightarrow \infty$, if $\frac{\mathrm{d^2}K}{\mathrm{d}t^2} < 0$, then we say the corresponding function $s(t)$ is decelerated growth speed. For example, $s(t) = \frac{1}{t}, s(t) = s_0$, etc. Specifically, considering science topics, from 2015 to 2020, deep learning is consistent with accelerated network growth, while channel capacity is consistent with decelerated network growth.

\section{Knowledge Boom Threshold}

Knowledge boom refers to the fact that KQI in certain disciplines suddenly inflects into a state of faster growth, although most of the time it increases linearly. In this section, we attempt to explore and reveal the knowledge boom threshold.

\subsection{Phenomenon}

The researchers said the 21st century is the century of life sciences.\cite{biology21st} One interesting point is that KQI in biology and medicine also entered a period of accelerated growth from then on (Fig. \ref{fig:boom}). In addition, we find somewhat nonlinear changes in chemistry, materials, engineering, and other disciplines. But surprisingly, there's nothing unusual about the number of papers in these areas at that time. 

To explain this phenomenon, we find that a cascade of failures in the network exists at a critical point similar to the percolation threshold.\cite{cascade} Inspired by this, we get a knowledge boom threshold 
\begin{equation}
    m \sim a \log{n+1},
\end{equation}
where $m$ is the average degree, $a$ is the required number of active neighbors for each node to be active and $n$ is graph size. The threshold means that stable connections between knowledge are established for flourish. More specifically, we define the two states of knowledge as: \textbf{inactive} and \textbf{active}. This threshold indicates that starting with some \textbf{active} knowledge, the common heuristic of a \textbf{active} knowledge can induce an \textbf{inactive} knowledge, and eventually almost all knowledge can be activated. The detailed proof process will be introduced in Section \ref{sec:boomproof}.

In the experiments of 311 (292 subdisciplines and 19 first-level disciplines) disciplines with different scales, we find that 40 disciplines have experienced a knowledge boom. After calculating the value $a$ of these disciplines, we find that different disciplines have different thresholds $a$, indicating that the difficulty to follow through from one paper to another is not the same. Interestingly, the disciplines of developing works (neuroscience, genetics, cognitive psychology, etc.) do not usually reach the threshold, whereas the disciplines of disruptive works (engineering, internet privacy, nanotechnology, etc.) do the opposite (Fig. \ref{fig:afield}). One reasonable explanation is that developing works are hard to comprehensively understand and have a higher $a$. For each discipline, as the scientific productivity increases, the knowledge amount has a phase transition when it reaches this threshold.

\begin{figure}[!htp]
  \centering
  \includegraphics[width=14.8cm]{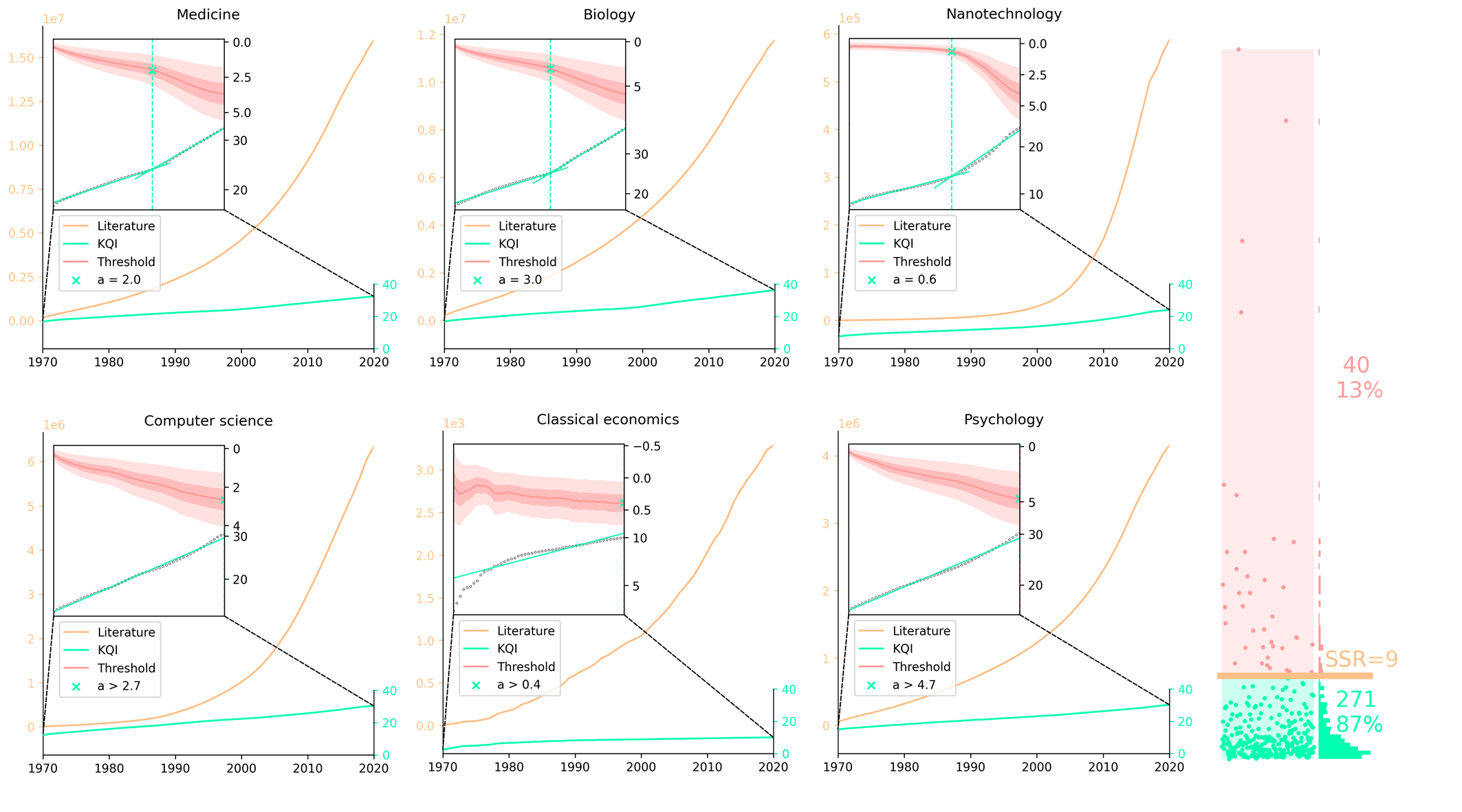} \\
  \caption[Growth of KQI.]{Growth of KQI. The growth rate of KQI was significantly different from that of the number of literature. Medicine, Biology, Nanotechnology are disciplines where knowledge booms have occurred, while Computer science, Classical economics, and Psychology are disciplines where knowledge booms have not yet begun. The threshold is calculated by formula $\frac{m-1}{\log{n}}$ (red line, i.e. $a$), $\frac{(1\pm20\%)m-1}{\log{n}}$ and $\frac{(1\pm50\%)m-1}{\log{n}}$ (shades of red), $a$ indicates the threshold of a discipline to accelerate the growth of knowledge. The KQI curve (cyan line) is fitted from real data points (black hollow circle) for each year. Among 311 (292 subdisciplines and 19 first-level disciplines) disciplines, about 13\% have boomed, with residual sum of squares (RSS=9) of linear fitting as the critical point.}
 \label{fig:boom}
\end{figure}

\begin{figure}[!htp]
  \centering
  \includegraphics[width=14.8cm]{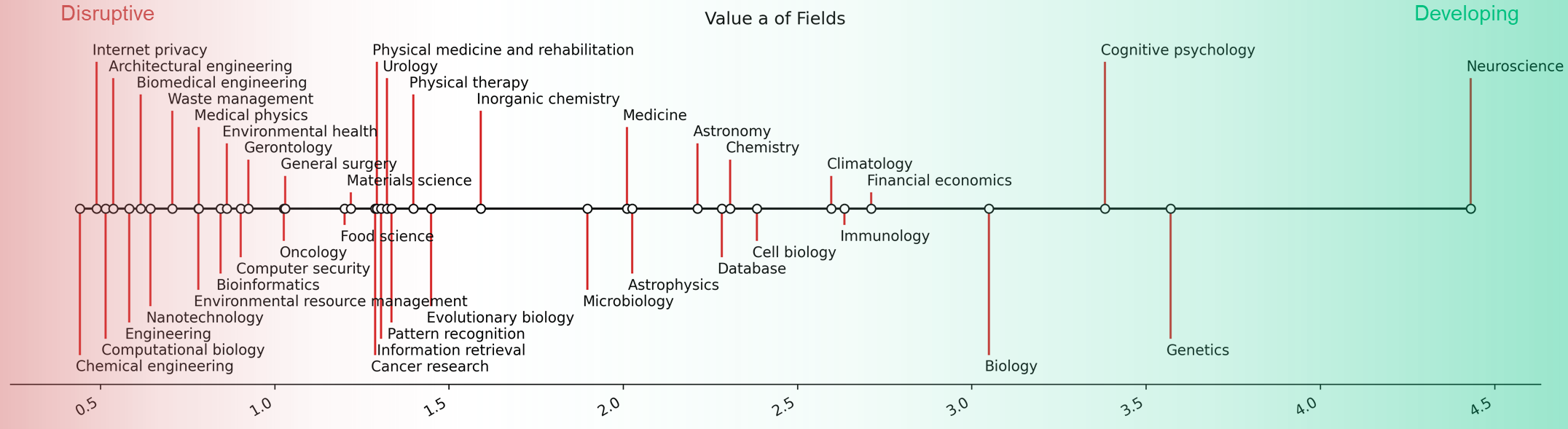} \\
  \caption[Knowledge boom threshold of different disciplines.]{Knowledge boom threshold of different disciplines. The disciplines of developing works do not usually reach the threshold (higher $a$), whereas the disciplines of disruptive works do the opposite (lower $a$).}
 \label{fig:afield}
\end{figure}

\subsection{Proof}
\label{sec:boomproof}
Intuitively, we think that a knowledge network to flourish requires a stable network of connections between all knowledge. In order to achieve this goal, the knowledge should be able to flow well in the network and inspire other knowledge. We think there are two states of knowledge in a network. One state is \textbf{active}, that is, the knowledge has been fully recognized. The other state is \textbf{inactive}, where the knowledge has not been fully recognized. An \textbf{active} knowledge can drive \textbf{inactive} knowledge neighbors, and \textbf{inactive} knowledge becomes active when it has $a$ neighbor of \textbf{active} knowledge.

Specifically, knowledge boom is identified as the following processes: (1) Randomly select some knowledge to be active. (2) Repeatedly search for inactive knowledge that has more than $a$ active neighbors, until no inactive knowledge can be activated. This process is similar to the learning process when we first enter a discipline. Starting from some known knowledge, we continue to expand our horizons and learn as much new knowledge as possible. 

Considering a network with an average degree $m$, one \textbf{active} knowledge will cause an average of $\frac{m-1}{a}$ \textbf{inactive} knowledge to be active, according to the conclusion of percolation process. Here, we use $m-1$ because the edge that causes the knowledge to be active should be excluded.

According to the relevant research on connectivity threshold\cite{connectivitythreshold}, starting from any knowledge in the network, the whole network becomes active when $$\frac{m-1}{a} > \log n.$$

That is $$m > a\log n + 1.$$

The proof is completed.

\section{80/20 Rule of KQI}
\label{sec:8020rule}
The Pareto rule, or 80/20 rule, used to be famous in economics for the fact that the top 20 percent of the population owned around 80 percent of the wealth. Exploring further, we find a similar 80/20 rule in knowledge, that is, 13.92\% of scientific productivity occupy 86.08\% of knowledge, and 86.08\% of scientific productivity occupy 13.92\% of knowledge (Fig. \ref{fig:80_20}). Unlike the contradictions between the rich and the poor, KQI reflects that the extraordinary papers are set off by tremendous ordinary papers. Without a large number of these ordinary papers, the extraordinary papers would not emerge. This phenomenon is pretty exciting, because we may only need to study a few papers, can obtain the vast majority of the knowledge of the entire academic network. 

\begin{figure}[!htp]
  \centering
  \includegraphics[width=7cm]{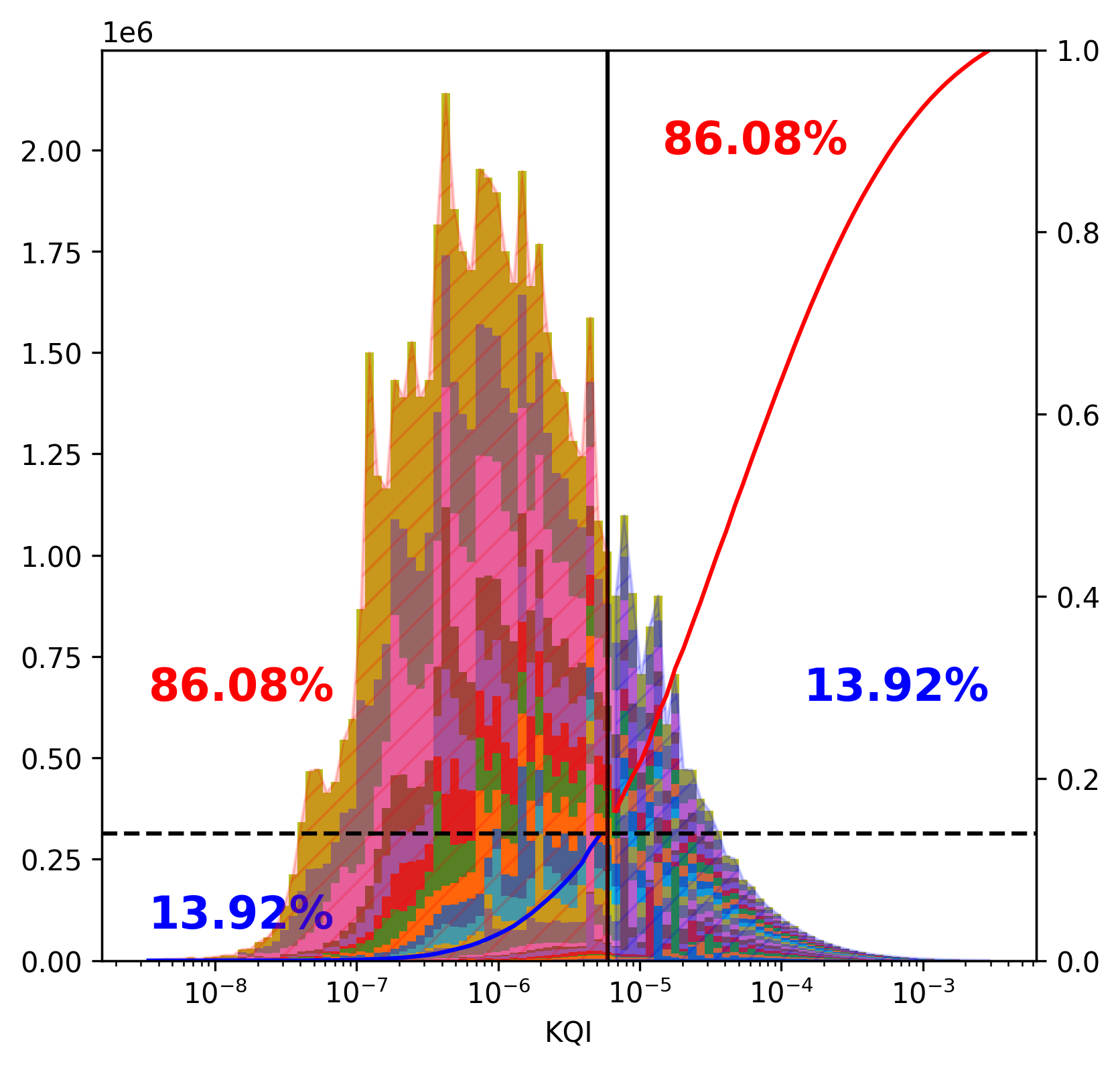} \\
  \caption[80/20 rule in knowledge.]{80/20 rule in knowledge. Stack distribution of KQI for all disciplines and cumulative distribution function are shown. 86.08\% of scientific productivity occupy 13.92\% of knowledge, and 13.92\% of scientific productivity occupy 86.08\% of knowledge.}
 \label{fig:80_20}
\end{figure}

It is important to note that while the 80/20 rule states that a small amount of literature represents the vast majority of knowledge, we should not ignore the value of those 80\% of literature. Low KQI does not mean that the literature is worthless, just that the literature can be summarized by that with high KQI. Literature with low KQI doesn't exist as indispensable nodes in the academic citation network but don't forget that the whole structure depends on the contribution of those 80\% of literature. The power of the masses is great. After removing a portion of literature with low KQI, we can find that some literature with high KQI is also fallen from the list.

% !TEX root = ../main.tex

\chapter{Applications of KQI}
\label{Chapter:application}

It can be seen from the previous chapter that KQI reveals some characteristics of the evolution and distribution of knowledge in academic citation networks. Beyond that, it has many further uses. In this chapter, we will give some practical applications of KQI, including extraction of knowledge veins, and rankings based on KQI.

\section{Knowledge Veins Extraction from Academic Networks}

To reveal the development of academic literature, we try to extract the key structure of disciplines. We propose the knowledge vein, just like the veins of leaves or blood vessels, and we hope that it can use the fewest papers to restore the most knowledge that reflects the development of an academic discipline through KQI. 
\begin{definition}[Knowledge Vein]
Knowledge vein defines as a compression of the academic citation network, which has the following characteristics: 
    \begin{enumerate}
        \item The knowledge vein is much smaller than the size of the academic citation network.
        \item Nodes in the knowledge vein are representatives (e.g. highest KQI) of the academic citation network.
        \item The knowledge vein completely preserves the inheritance relationship of representatives, that is, the nodes in the knowledge vein that can be connected must be connected in the academic citation network, and vice versa.
    \end{enumerate}
\end{definition}

This idea is inspired by the 80/20 rule mentioned in Section \ref{sec:8020rule}, since a small number of papers can represent the majority of knowledge, which means that we don't have to read all the papers.

After getting a high KQI literature list, we compress the original academic citation network. The list of literature with a high KQI corresponds to a list of nodes in the graph. For any two nodes, if there is a path between them, and no other high KQI nodes are passed along any paths, the two nodes are connected by an edge. The specific method of obtaining knowledge veins is shown in Algorithm \ref{algorithm:vein}.

\begin{algorithm}[htb]
\caption{Knowledge veins extraction.}
\label{algorithm:vein}
% \small
\SetAlgoLined
\DontPrintSemicolon

\SetKwInput{KwInput}{Input}
\SetKwInput{KwOutput}{Output}

    \KwInput{original directed acyclic graph $G$}
    \KwOutput{knowledge veins $G'$}
    
    Initializes a set of nodes with high KQI to nodelist.\;
    \For{node $\in$ nodelist}{
        $G'$.addnode(node)\;
        maxdepth := 1\;
        \While{$G'$.indegree(node) == 0}{
            openlist := [(v, 1) for v in $G$.predecessors(node)]\;
            accesslist := [v for v in $G$.predecessors(node)]\;
            \While{openlist is not empty}{
                v, depth := openlist.pop()\;
                \If{v $\in$ nodelist}{
                    $G'$.addedge(v, node)\;
                    continue\;
                }
                \If{depth < maxdepth}{
                    \For{s $\in$ $G$.predecessors(v)}{
                        \If{s $\notin$ accesslist}{
                            openlist.append((s, depth+1))\;
                            accesslist.add(s)\;
                        }
                    }
                }
            }
            maxdepth := maxdepth+1\;
            \If{maxdepth > MAXDEPTH}{break}
        }
    }
\end{algorithm}

Considering the inheritance structure of knowledge, in order to better visualize the literature, we arrange it according to the citation relations (root on top, leaf on bottom). Taking data mining as an example, we select the papers with the highest ranking in KQI to generate the veins that are preserved to varying degrees. Visible, the quality of the selected articles is guaranteed, and they are scattered around  (Fig. \ref{fig:datamining}). It shows that the papers with the highest rankings in KQI not only cover the majority of knowledge but also represent the development of the discipline well. We also do experiments in many other different disciplines, like business, astronomy, mathematical analysis, nanotechnology, cancer research, and so on, and we get the same results (Fig. \ref{fig:veins}). In this way, any new literature can be quickly associated with the existing knowledge system. In addition, the knowledge veins will change with the development of the discipline and the transfer of hot spots.

The extracted knowledge veins can be useful, for example, in helping an entry-level researcher decide which articles to read, or in helping a cross-disciplinary worker quickly grasp the essence of a new discipline. The knowledge veins can also help researchers to establish a good knowledge system, construct the concept of discipline development, and write literature reviews.

\begin{figure}[!htp]
  \centering
  \includegraphics[width=14.8cm]{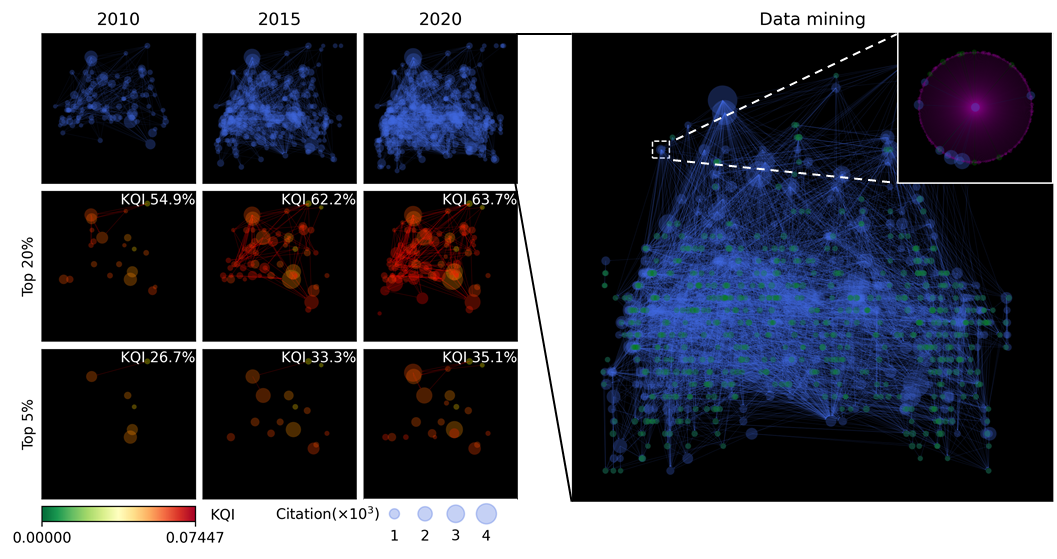} \\
  \caption[The structural vein of academic citation network of data mining.]{The structural vein of academic citation network of data mining. Here papers with citations greater than 1000 are selected as representatives. The first row is the overall picture in 2010, 2015, and 2020 (using dot\cite{dot} algorithm, drawing directed graphs as hierarchies). The second and third rows are drawn by selecting the top 20\% and 5\% literature in KQI respectively, and the content of KQI represented by the extracted structure is marked in the upper right corner. In the enlarged view of the original picture, green ones represent the papers with citations between 300 and 1000, and purple ones represent all papers.}
 \label{fig:datamining}
\end{figure}

\begin{figure}[!htp]
  \centering
  \includegraphics[width=10cm]{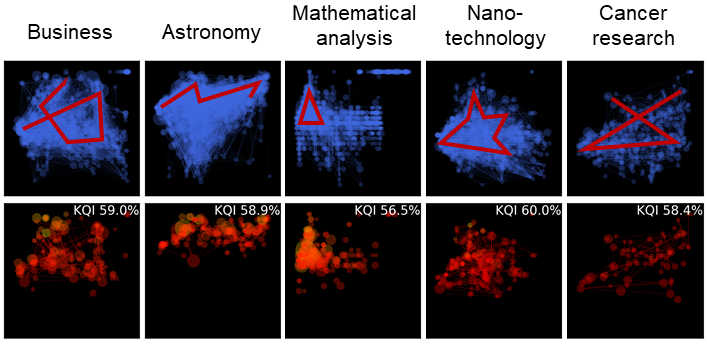} \\
  \caption[Typical veins.]{Typical veins. The selected articles with high KQI can help reflect the knowledge veins of the original domain, using less literature to cover more knowledge. The red lines are artificially marked to compare the similarities of the structures.}
 \label{fig:veins}
\end{figure}

\section{KQI-based Rankings}

We have known the KQI of every single paper, to some extent, reflects the knowledge attribute of the paper, as mentioned acceptability and dependability above (Fig. \ref{fig:dependability_acceptability}). The higher the value of KQI, the stronger the attribute of these two items of knowledge. Therefore, we try to make a ranking for papers, which can reflect the above characteristics of the papers. Then based on the ranking of papers, we also rank the authors, affiliations, and countries to further confirm the effect of our proposed KQI.

\begin{figure}[!htp]
  \centering
  \includegraphics[width=5cm]{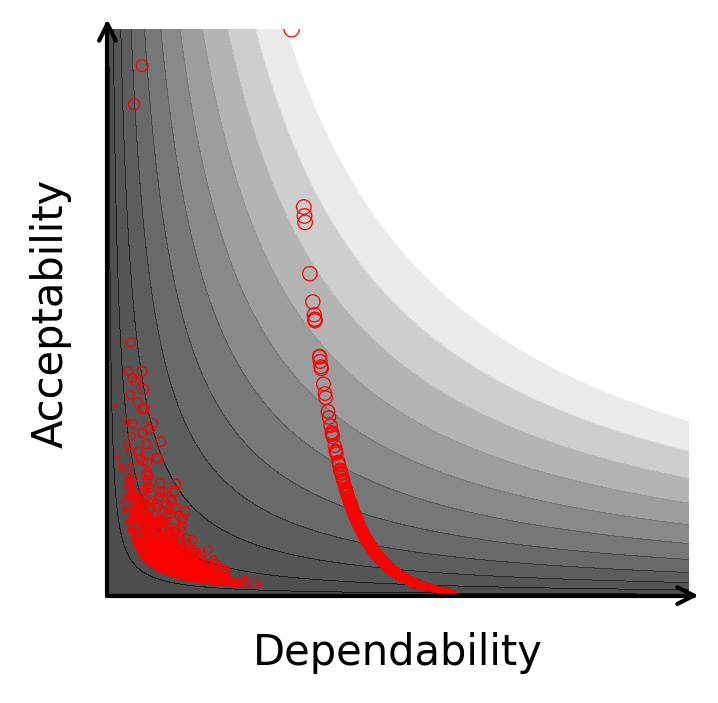} \\
  \caption[Acceptability and dependability of KQI.]{Acceptability and Dependability of KQI. A paper with a high KQI should strike a good balance between these two properties.}
 \label{fig:dependability_acceptability}
\end{figure}

\subsection{Discovery of Valuable Literature}

Taking the computer science discipline as an example, we select the top 5 papers in KQI and compared their KQI with citations year by year. They are \textit{Gradient-based learning applied to document recognition}, \textit{Very deep convolutional networks for large-scale image recognition}, \textit{Principles of neurodynamics}, \textit{Parallel rough set based knowledge acquisition using MapReduce from big data} and \textit{Pattern classification and scene analysis}. All five papers have been manually reviewed and are considered influential in the field of computer science. We can find that the KQI of a paper changes with the shift of research hotspots, which is consistent with our intuition (Fig. \ref{fig:paper}). When the hot spots move away, the KQI of the paper may decrease. For example, the research related to neural networks originated in the 1980s, but was eclipsed by the emergence of algorithms such as SVM around 1995, and flourished with deep learning in recent years. This is evident in the changes of KQI but is not reflected in the number of citations. 

\begin{figure}[!htp]
  \centering
  \includegraphics[width=10cm]{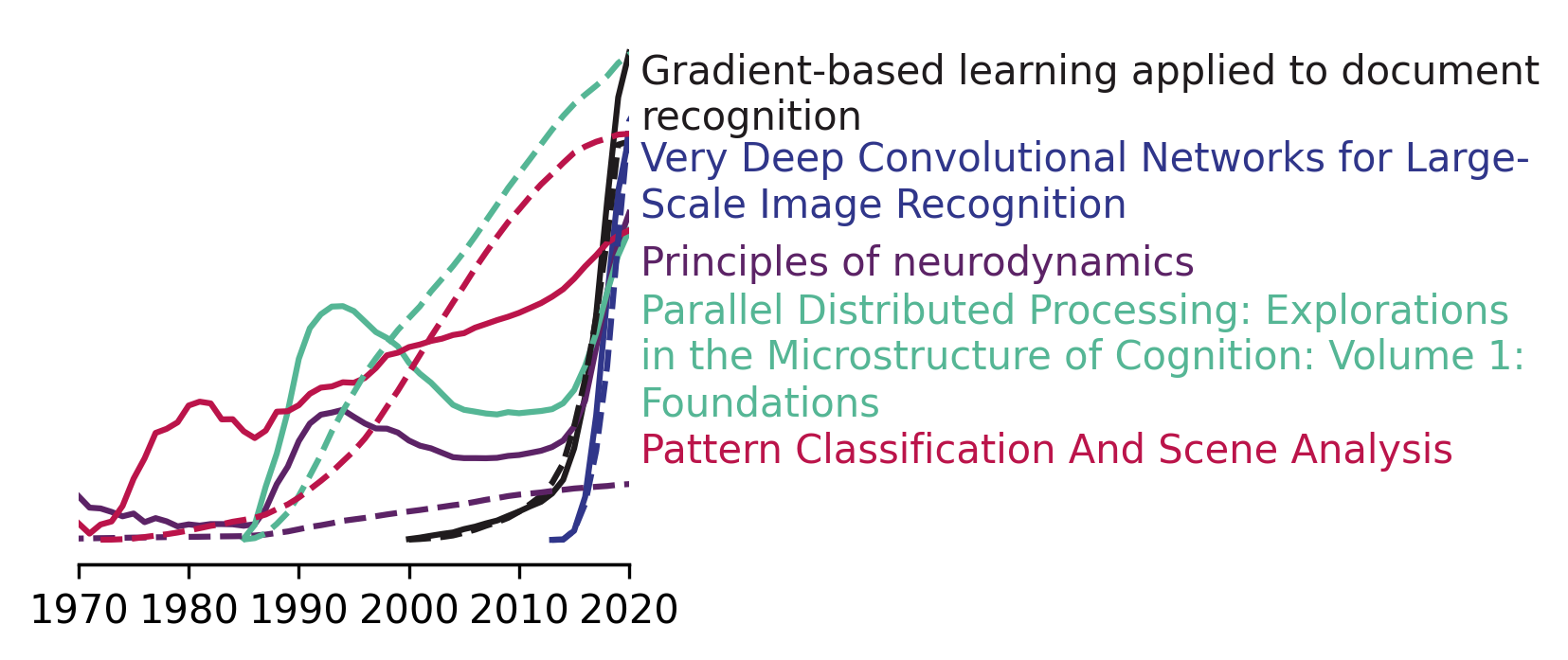} \\
  \caption[KQI versus citations.]{KQI versus citations. The top five papers of KQI in computer science are selected, and the changes of their KQI (solid) and citations (dotted) over time were mapped. Top-ranked papers do seem to have considerable value, and the KQI is related to the number of citations, representing the amount of knowledge value of the paper at the time.}
 \label{fig:paper}
\end{figure}

In addition, KQI can find valuable papers that are not highly cited. For example, in a ranking of the discipline of network coding, the top one is \textit{Network multicast flow via linear coding} published in 1998 on International Symposium on Operations Research and its Applications in engineering, technology, and management (ISORA), despite its low citation. It is verified that this paper ranks first because it has produced some high-impact papers, such as \textit{Linear network coding}, \textit{An algebraic approach to network coding} and \textit{Network coding theory}, which are famous in the field of network coding. Because high-impact papers all cite this paper, it appears to be of considerable value.

To sum up, KQI is essentially a citation-based metric, but different from citations. Specifically, there are four differences between them:
\begin{itemize}
    \item \textbf{Citations only focus on the number of references to one paper. But, KQI considers the citations of all derived papers.} This is the benefit brought by the introduction of volume. Even if the citation of a paper is not high, if it produces some influential papers, it will indirectly explain the value of this paper.
    \item \textbf{Citations will only increase, and the influence of the paper will only grow. But, KQI will increase or decrease dynamically with time, which reflects the knowledge content of the paper at a specific time.} This can be seen in the results of the experiment. When the graph structure is changing, the value of a certain paper will change accordingly, which may increase or decrease, affected by the latest research hotspots.
    \item \textbf{Citations only reflect the truth of a paper. KQI considers the truth and belief of a paper.} As mentioned in Section \ref{sec:explainKQI}, JTB theory is reflected  on KQI. Truth is reflected by citations because it is widely accepted. Belief is reflected by papers' references, which tells if the paper is based on some widely accepted paper. The KQI's formula takes into account both of these two.
    \item \textbf{Citations are easy to attack, and manipulating KQI is even more difficult.} Because the KQI depends on the structure of the entire citation network unless you have the power to break it, it will not do much damage to the KQI. KQI is slightly better at resisting attacks than citations.
\end{itemize}

If we draw a matrix corresponding relationship between knowledge and citations, we can find there are four states: \textbf{high} knowledge and \textbf{high} citations, \textbf{low} knowledge and \textbf{low} citations, \textbf{high} knowledge and \textbf{low} citations, \textbf{low} knowledge and \textbf{high} citations. Citations are qualified to the first two states, and KQI complements the latter two. For example, \textit{Knowledge matters: importance of prior information for optimization} written by Yoshua Bengio receives low citations and high KQI, however \textit{Top 10 algorithms in data mining} receives high citations and relatively low KQI. The experiments are carried out in 19 disciplines, and the top 3 of each discipline are listed in Table \ref{tab:top3}.

\vspace{7mm}

\begin{longtable}[!htp]{llr}
\caption{KQI Top@3 papers in 19 major first-level disciplines.}
\label{tab:top3} \\
\toprule
    \textbf{Field} & \textbf{Title} & \textbf{KQI} \\
\midrule
\endfirsthead
\multicolumn{3}{r}{Table \thetable (continued).} \\
\toprule
    \textbf{Field} & \textbf{Title} & \textbf{KQI} \\
\midrule
\endhead
\hline
% \multicolumn{3}{r}{续下页}
\endfoot
\endlastfoot
	\multirow{3}{2cm}{Art} & Visual Pleasure and Narrative Cinema & 0.0147 \\
	& The Work of Art in the Age of Mechanical Reproduction & 0.0078 \\
	& The Language of New Media & 0.0065 \\
\hline
	\multirow{3}{2cm}{Biology} & Molecular cloning : a laboratory manual & 0.0405 \\
	& Atlas of protein sequence and structure & 0.0246 \\
	& DNA sequencing with chain-terminating inhibitors & 0.0167 \\
\hline
	\multirow{3}{2cm}{Business} & A Behavioral Theory of the Firm & 0.0531 \\
	& Organizations in Action & 0.0406 \\
	& The theory of buyer behavior & 0.0321 \\
\hline
	\multirow{3}{2cm}{Chemistry} & Protein Measurement with the Folin Phenol Reagent & 0.0217 \\
	& CRC Handbook of Chemistry and Physics & 0.0112 \\
	& Colorimetric method for determination of sugars and related substances & 0.0087 \\
\hline
	\multirow{3}{2cm}{Computer science} & Gradient-based learning applied to document recognition & 0.04 \\
	& Very Deep Convolutional Networks for Large-Scale Image Recognition & 0.0331 \\
	& Principles of neurodynamics & 0.0268 \\
\hline
	\multirow{3}{2cm}{Economics} & CAPITAL ASSET PRICES: A THEORY OF MARKET ... & 0.0226 \\
	& A New Approach to Consumer Theory & 0.0184 \\
	& An Evolutionary Theory of Economic Change & 0.0183 \\
\hline
	\multirow{3}{2cm}{Engineering} & Selected harmonic reduction in static D-C - A-C inverters & 0.0226 \\
	& A New Neutral-Point-Clamped PWM Inverter & 0.0198 \\
	& Fundamentals of Heat and Mass Transfer & 0.0181 \\
\hline
	\multirow{3}{2cm}{Environmental science} & Standard methods for the examination of water and wastewater & 0.0643 \\
	& Climate Change 2007: The Physical Science Basis & 0.045 \\
	& Climate Change 2001: The Scientific Basis & 0.0419 \\
\hline
	\multirow{3}{2cm}{Geography} & Climate Change 2007: Impacts, Adaptation and Vulnerability. ... & 0.0215 \\
	& Remote Sensing and Image Interpretation & 0.0184 \\
	& Climate change 2007 : the physical science basis : contribution of ... & 0.0175 \\
\hline
	\multirow{3}{2cm}{Geology} & The calculation of the directional reflectance of a vegetative canopy & 0.0158 \\
	& River flow forecasting through conceptual models part I - A ... & 0.0154 \\
	& Free software helps map and display data & 0.0117 \\
\hline
	\multirow{3}{2cm}{History} & Imperial Eyes: Travel Writing and Transculturation & 0.0067 \\
	& The Climate of History: Four Theses & 0.0058 \\
	& Provincializing Europe: Postcolonial Thought and Historical Difference & 0.0053 \\
\hline
	\multirow{3}{2cm}{Materials science} & Physics of graphite & 0.0306 \\
	& Science of fullerenes and carbon nanotubes & 0.0197 \\
	& Theory of elasticity & 0.0172 \\
\hline
	\multirow{3}{2cm}{Mathematics} & An introduction to probability theory and its applications & 0.0209 \\
	& Matrix Computations & 0.0156 \\
	& Partial differential equations & 0.0151 \\
\hline
	\multirow{3}{2cm}{Medicine} & Design and analysis of randomized clinical trials requiring prolonged ... & 0.0089 \\
	& Summary of the second report of the National Cholesterol Education ... & 0.0072 \\
	& Statistical Aspects of the Analysis of Data From Retrospective ... & 0.0069 \\
\hline
	\multirow{3}{2cm}{Philosophy} & Knowledge and Its Limits & 0.0204 \\
	& A treatise of human nature & 0.02 \\
	& Counterpart Theory and Quantified Modal Logic & 0.0174 \\
\hline
	\multirow{3}{2cm}{Physics} & Quantum Computation and Quantum Information & 0.0248 \\
	& Principles of Optics & 0.0136 \\
	& Quantum Mechanics & 0.012 \\
\hline
	\multirow{3}{2cm}{Political science} & An Economic Theory of Democracy & 0.0142 \\
	& The Nature and Origins of Mass Opinion & 0.0097 \\
	& Making Democracy Work: Civic Traditions in Modern Italy & 0.009 \\
\hline
	\multirow{3}{2cm}{Psychology} & Diagnostic and Statistical Manual of Mental Disorders & 0.0378 \\
	& An inventory for measuring depression & 0.0188 \\
	& A RATING SCALE FOR DEPRESSION & 0.017 \\
\hline
	\multirow{3}{2cm}{Sociology} & The Structure of Scientific Revolutions & 0.0193 \\
	& The Discovery of Grounded Theory & 0.0188 \\
	& Distinction: A Social Critique of the Judgement of Taste & 0.0181 \\
\bottomrule
\end{longtable}

\subsection{Authors Ranking}

To make the results more convincing, we aggregate the KQI of the papers by author, and then give a ranking of the authors. Due to the additivity of entropy, our aggregation is just to add up KQI of all the papers of an author. We do experiments in few different disciplines, and interestingly, the authors at the top of the list are actually quite influential. In order to eliminate subjective factors, the Turing Award and Nobel Prize, which are highly recognized in the academic circle, are selected as the evaluation criteria. 

In our experimental dataset, there are 6,253,122 authors in computer science. By 2020, there are 74 Turing Award winners. As shown in Table \ref{tab:topauthor}, thirty percent of the top 50 authors according to KQI are Turing Award winners, while the remaining 70 authors are also highly influential and receive honors such as the IEEE John von Neumann Medal, MacArthur Fellows Program, or MacArthur Fellowship, Frederick W. Lanchester Prize, etc. We single out the top 10,000 authors (0.16\%) according to KQI and find 71 Turing Award winners (96\%) among them (Fig. \ref{fig:turingnobel}a). The remaining three authors are Alan Perlis, James H. Wilkinson, and Kristen Nygaard. In the case of Alan Perlis, due to some objective factors, we do not have his representative works in our dataset. For James H. Wilkinson, his outstanding contribution in numerical analysis is classified into the field of mathematics by our dataset. Kristen Nygaard, who co-invented object-oriented programming and the Simula programming language with Ole-Johan Dahl, is listed as the second author. Kristen Nygaard falls behind in the rankings because we only consider the first author when dealing with the rankings. When we use the h-index to rank authors, it is far from being able to achieve this effect. Many Turing Award winners do not have a high h-index, such as Edwin Catmull, Raj Reddy, Ken Thompson, and so on.

\begin{table}[!hpt]
  \caption{Top@50 authors in computer science.}
  \label{tab:topauthor}
  \centering
  \begin{tabular}{@{}rclr|rclr@{}} 
    \toprule
    No & Author & KQI & Turing & No & Author & KQI & Turing\\ 
    \midrule
    1 & M. R. Garey & 0.04292 &  & 26 & Niklaus Wirth & 0.01695 & 1984 \\
    2 & Gerard Salton & 0.04213 &  & 27 & Vladimir Vapnik & 0.01820 &  \\
    3 & Alfred V. Aho & 0.04185 & 2020 & 28 & Bernhard Schölkopf & 0.01731 &  \\
    4 & Claude E. Shannon & 0.04118 &  & 29 & Lotfi A. Zadeh & 0.01703 &  \\
    5 & Edsger W. Dijkstra & 0.03852 & 1972 & 30 & Ronald L. Rivest & 0.01687 & 2002 \\
    6 & Donald Ervin Knuth & 0.02891 & 1974 & 31 & Maurice V. Wilkes & 0.01623 & 1967 \\
    7 & Teuvo Kohonen & 0.03595 &  & 32 & Ian T. Foster & 0.01640 &  \\
    8 & Richard O. Duda & 0.03175 &  & 33 & Bernard Widrow & 0.01637 &  \\
    9 & Dimitri P. Bertsekas & 0.02895 &  & 34 & E. F. Codd & 0.01622 & 1981 \\
    10 & C. A. R. Hoare & 0.02587 & 1980 & 35 & Richard M. Karp & 0.01593 & 1985 \\
    11 & Geoffrey E. Hinton & 0.02050 & 2018 & 36 & Judea Pearl & 0.01571 & 2011 \\
    12 & David E. Rumelhart & 0.02493 &  & 37 & Rakesh Agrawal & 0.01579 &  \\
    13 & Leonard Kleinrock & 0.02189 &  & 38 & Allen Newell & 0.01560 & 1975 \\
    14 & Thomas H. Cormen & 0.02103 &  & 39 & John E. Hopcroft & 0.01528 & 1986 \\
    15 & Marvin Minsky & 0.02019 & 1969 & 40 & Michael Stonebraker & 0.01441 & 2014 \\
    16 & Lawrence R. Rabiner & 0.02049 &  & 41 & Tim Berners-Lee & 0.01340 & 2016 \\
    17 & John McCarthy & 0.02011 & 1971 & 42 & Noam Chomsky & 0.01482 &  \\
    18 & Yann LeCun & 0.01825 & 2018 & 43 & Jim Gray & 0.01297 & 1998 \\
    19 & Richard O. Duda & 0.02004 &  & 44 & Zohar Manna & 0.01414 &  \\
    20 & Azriel Rosenfeld & 0.01988 &  & 45 & Erich Gamma & 0.01393 &  \\
    21 & Frederick Jelinek & 0.01978 &  & 46 & William H. Press & 0.01392 &  \\
    22 & Thorsten Joachims & 0.01928 &  & 47 & Ben Shneiderman & 0.01379 &  \\
    23 & John H. Holland & 0.01915 &  & 48 & A. A. Mullin & 0.01376 &  \\
    24 & David E. Goldberg & 0.01878 &  & 49 & Gene H. Golub & 0.01356 &  \\
    25 & John G. Proakis & 0.01834 &  & 50 & Whitfield Diffie & 0.01295 & 2015 \\
    \bottomrule
  \end{tabular}
\end{table}

We also conduct experiments in other disciplines, such as economics, find that KQI also singles out 85 (98\%) of the 86 winners of the Nobel Memorial Prize in Economic Sciences (Fig. \ref{fig:turingnobel}b). The only author not on the list, Leonid Vitaliyevich Kantorovich, is known for linear programming, which is classified by our dataset in the mathematical field. Therefore, the author's KQI ranking does a good job of screening out the Turing Award and Nobel prize winners, which the h-index cannot do. In fact, we just take these two well-known awards for example. There are also many famous authors ranked top by KQI, such as the father of information retrieval Gerard Salton, the father of information theory Claude Shannon, and so on.

\begin{figure}[!hbtp]
    \centering
    \subcaptionbox{KQI of Turing Award winners.}[14.8cm]{
        \includegraphics[width=14.8cm]{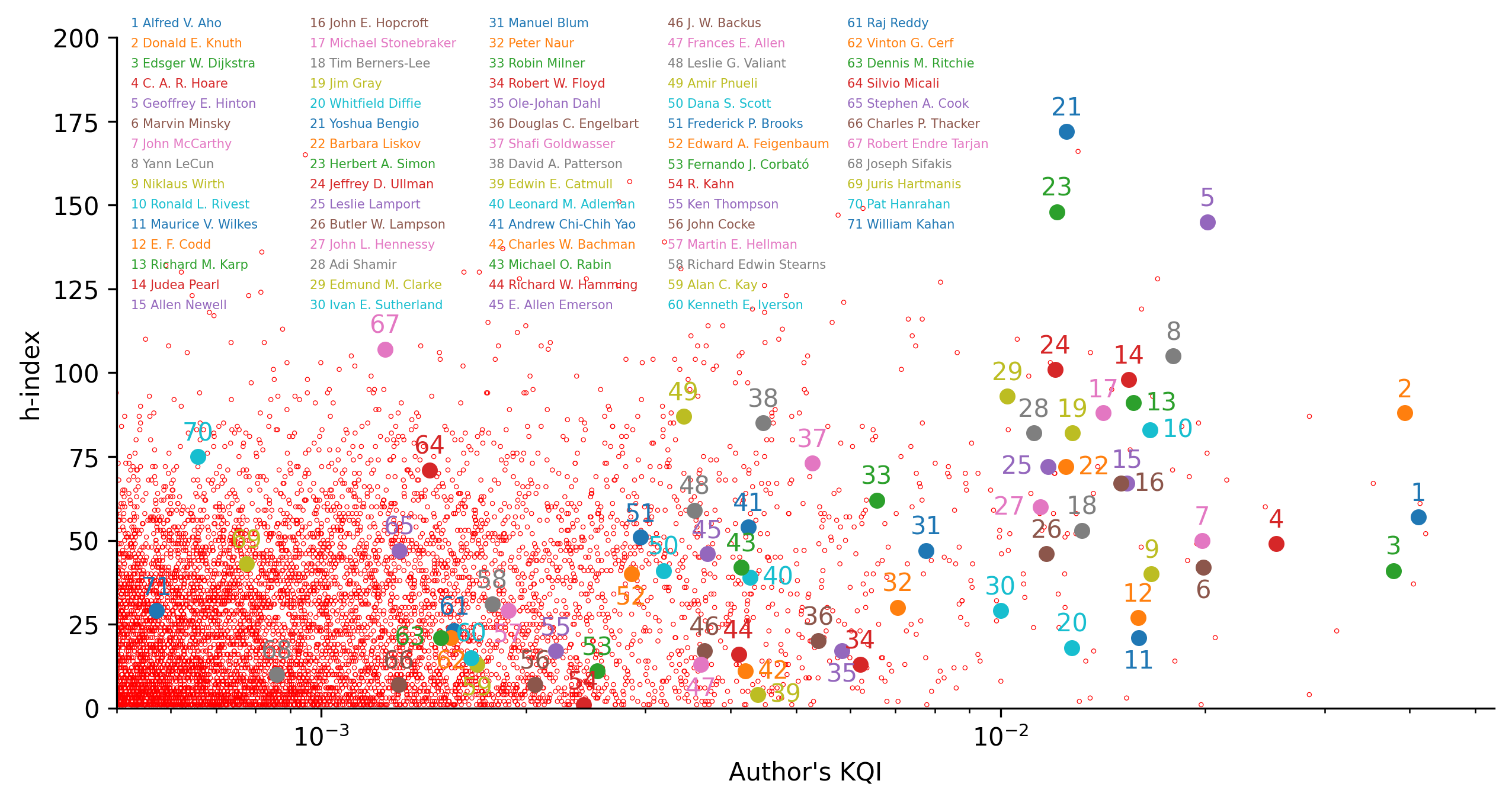}
    }
    \subcaptionbox{KQI of Nobel Prize winners.}[14.8cm]{
        \includegraphics[width=14.8cm]{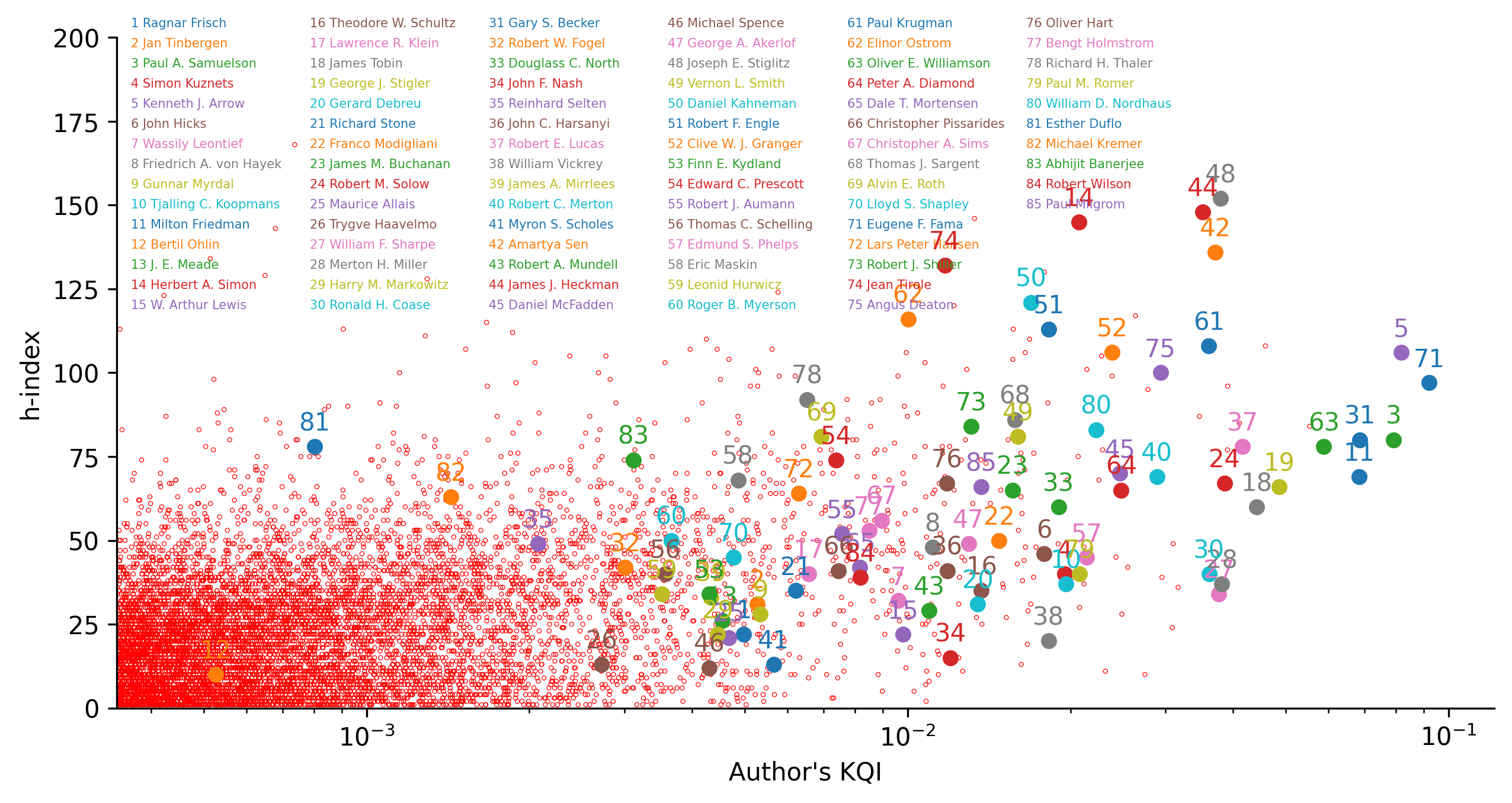}
    }
    \caption[KQI of Turing Award and Nobel Prize winners.]{KQI of Turing Award and Nobel Prize winners. By zooming in on the part of high KQI (top 10,000 authors), 71 of the 74 Turing Awards and 85 of the 86 Nobel Prizes in Economic Sciences by 2020 are included, which is far beyond the competence of the h-index.}
    \label{fig:turingnobel}
\end{figure}

\subsection{Affiliations and Countries Rankings}

Similar to author rankings, we aggregate the KQI of papers by affiliations and countries and list the top 20 affiliations and countries respectively.

Table \ref{tab:affiliation_rank} shows that Harvard, Stanford, NIH, MIT and UC Berkeley dominate the list. Among affiliations with high KQI, the top 20 are not with China. The top five Chinese affiliations are Chinese Academy of Sciences (34), Tsinghua University (175), Zhejiang University (225), Peking University (232) and Shanghai Jiao Tong University (260). The number of literature and KQI of the United States both far exceed those of other countries. Nowadays, China has almost half as much literature as the United States, but still lags far behind in the KQI (\ref{fig:country_rank}). This is also in response to a shift in China's scientific research in recent years from quantity to quality.

\begin{table}[!hpt]
  \caption{Affiliations ranking.}
  \label{tab:affiliation_rank}
  \centering
  \begin{tabular}{@{}rlrr@{}} 
    \toprule
    No & Affiliation & KQI & Literature\\
    \midrule
    1 & Harvard University & 0.43141 & 484234 \\
    2 & Stanford University & 0.30732 & 313713 \\
    3 & National Institutes of Health & 0.23750 & 239960 \\
    4 & Massachusetts Institute of Technology & 0.22077 & 236888 \\
    5 & University of California, Berkeley & 0.21931 & 277942 \\
    6 & University of Michigan & 0.18456 & 334403 \\
    7 & University of Cambridge & 0.18224 & 274418 \\
    8 & Columbia University & 0.17989 & 220260 \\
    9 & University of California, Los Angeles & 0.17820 & 268271 \\
    10 & University of Washington & 0.17401 & 272258 \\
    11 & Yale University & 0.15899 & 213039 \\
    12 & Max Planck Society & 0.15635 & 391754 \\
    13 & University of Chicago & 0.14993 & 181715 \\
    14 & University of Pennsylvania & 0.14967 & 239943 \\
    15 & University of Wisconsin-Madison & 0.14093 & 232613 \\
    16 & Cornell University & 0.14011 & 224910 \\
    17 & University of California, San Diego & 0.13731 & 188257 \\
    18 & University of Minnesota & 0.13681 & 274059 \\
    19 & Johns Hopkins University & 0.13138 & 209663 \\
    20 & University of Oxford & 0.13092 & 275976 \\
    \hline
    \hline
    34 & Chinese Academy of Sciences & 0.09036 & 522360 \\
    175 & Tsinghua University & 0.02592 & 165012 \\
    225 & Zhejiang University & 0.02159 & 163465 \\
    232 & Peking University & 0.02110 & 148313 \\
    260 & Shanghai Jiao Tong University & 0.01865 & 167006 \\
    299 & Fudan University & 0.01602 & 101690 \\
    325 & University of Science and Technology of China & 0.01403 & 77018 \\
    343 & Nanjing University & 0.01318 & 92568 \\
    367 & Huazhong University of Science and Technology & 0.01213 & 114264 \\
    390 & Sun Yat-sen University & 0.01145 & 97154 \\
    \bottomrule
  \end{tabular}
\end{table}

\begin{figure}[!htp]
    \centering
    \includegraphics[height=8cm]{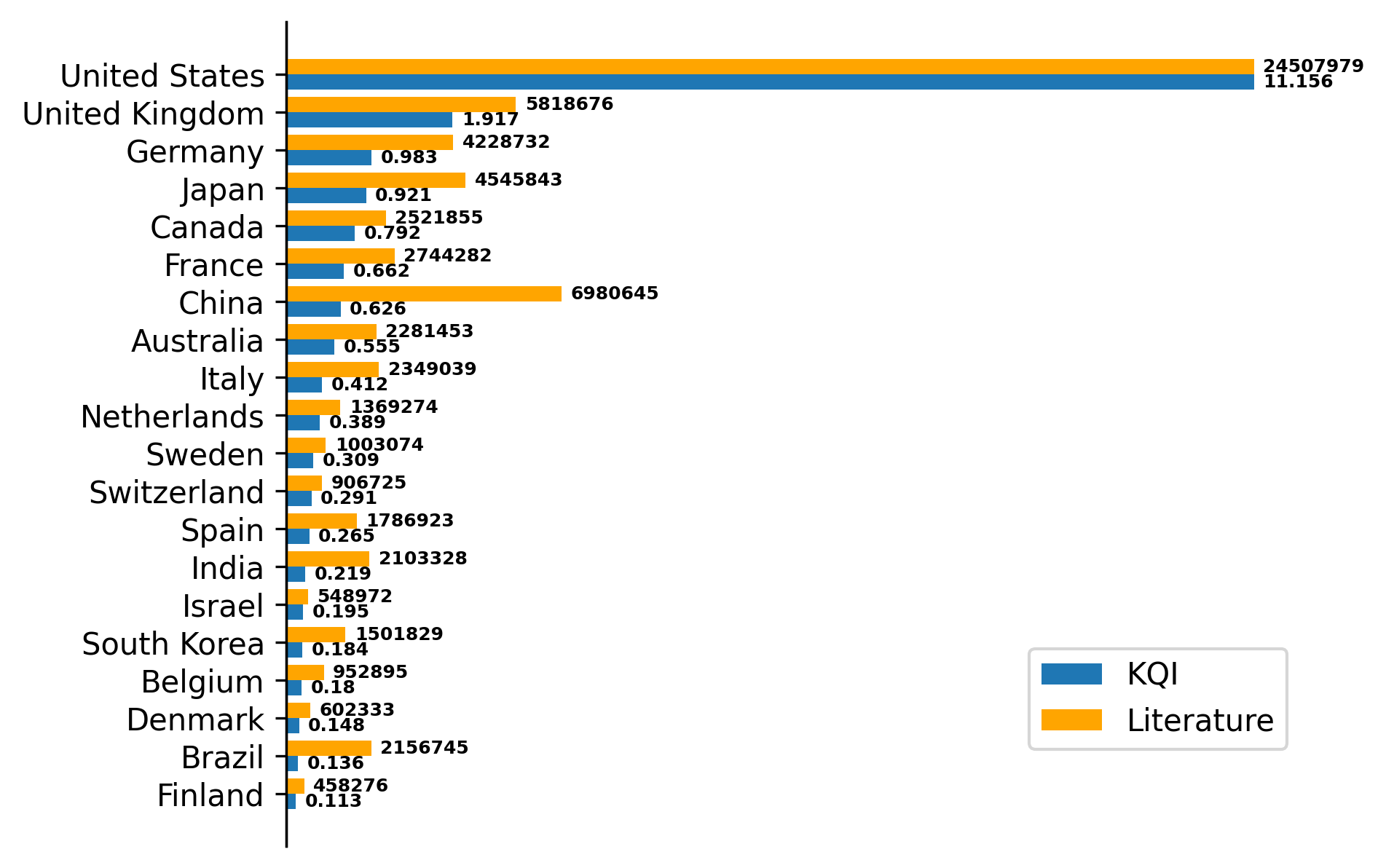} \\
    \caption{Countries ranking.}
    \label{fig:country_rank}
\end{figure}

\section{KQI vs. Other Classic Metrics}

From the above ranking, depending on the artificial experience, it can be considered that KQI has certain effectiveness. In order to further demonstrate the differences between our proposed KQI and other metrics, we randomly sample 10,000 papers, 10,000 authors, and 10,000 journals from the dataset, and draw a scatter plot of their KQI with PageRank, h-index, and impact factor respectively.

It is obvious that PageRank and KQI are positively correlated, but they're not exactly the same (Fig. \ref{fig:comparison}a). As a traditional method to measure the importance of nodes in a graph, PageRank performs a random walk on the graph to rank nodes by their information flow.  However, this method only tells the popularity, which is unequal to the knowledge. We are more interested in the quantity, value, and minimum redundancy of knowledge. To be more specific, KQI is better than PageRank in the following aspects:
\begin{itemize}
    \item \textbf{Interpretability.} PageRank is just a state of balanced information flow, which expresses influence and lacks interpretation at the knowledge level. KQI expresses the structure reflected by the difference between Shannon entropy and structural entropy, which is related to the meaning of knowledge.
    \item \textbf{Formulation.} PageRank can be viewed as a subset of KQI, that is, PageRank expresses a similar meaning with the volume variable $V$ mentioned in Eq. \ref{eq:KQI}. In this sense, KQI is more advanced than PageRank.
    \item \textbf{Complexity.} The algorithm complexity of PageRank depends on the number of iterations required to achieve convergence, while KQI only needs to traverse every node in the graph once in the preparation stage, and then KQI can be calculated with a constant complexity, and the algorithm complexity is more stable.
    \item \textbf{Additivity.} KQI bases on entropy, the difference between Shannon entropy and structural entropy, so KQI inherits the additivity of entropy while PageRank does not. For nodes in the network, KQI can be aggregated by summing up by any combination.
\end{itemize}

Some works have proposed that the role of h-index is equivalent to citations\cite{hindexbad} and still bound by citations, although h-index measures both the productivity and citation impact of a scientist. We find a weak correlation between the h-index and KQI. For high h-index scholars, KQI is usually not too bad (Fig. \ref{fig:comparison}b). However, compared with their h-index, KQI corrects those scholars who exploit the h-index loophole to a certain extent. Besides, h-index often buries some outstanding scholars, such as the mentioned Turing Award and Nobel Prize winners, which are included in KQI.

Also, some works questioned impact factor for abuse\cite{rethinkingIF}, although impact factor is frequently used as a proxy for the relative importance of a journal within its field. Our experiment confirms that the impact factors of journals have a limited role in determining the value of their published papers, which fits our intuition. It can only be inferred that journals with larger impact factors are less likely to receive bad articles, but it cannot be inferred that journals with smaller impact factors have no valuable articles (Fig. \ref{fig:comparison}c). This has considerable guiding significance for us. There is no need to be obsessed with authoritative journals. The quality of the articles should not be evaluated directly by the level of the journal, but by the value of the article itself.

\begin{figure}[!hbtp]
    \centering
    \subcaptionbox{Comparison with PageRank.}[4.5cm]{
        \includegraphics[height=4cm]{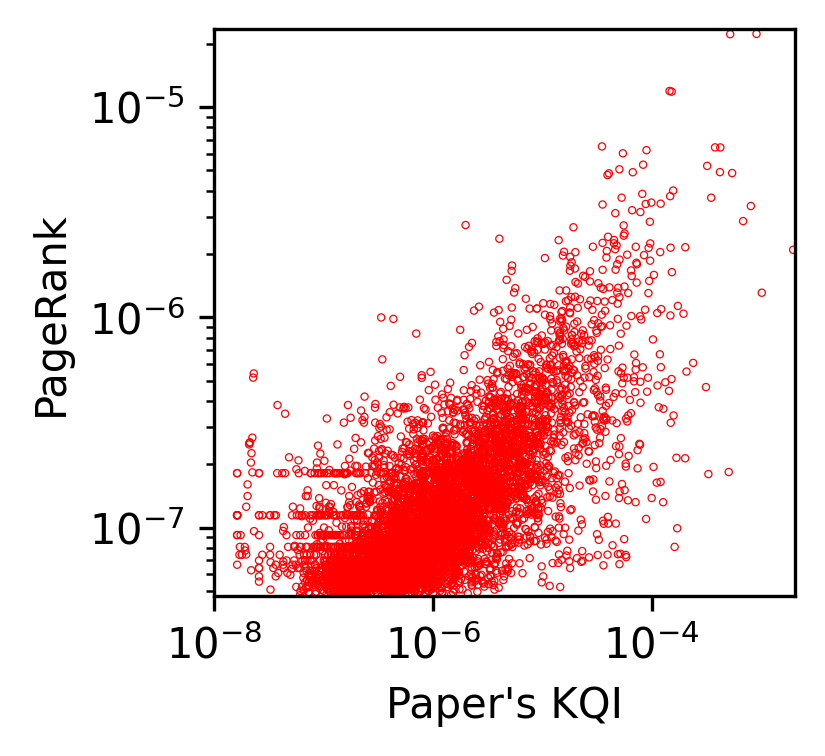}
    }
    \hspace{0.2cm}
    \subcaptionbox{Comparison with h-index.}[4.5cm]{
        \includegraphics[height=4cm]{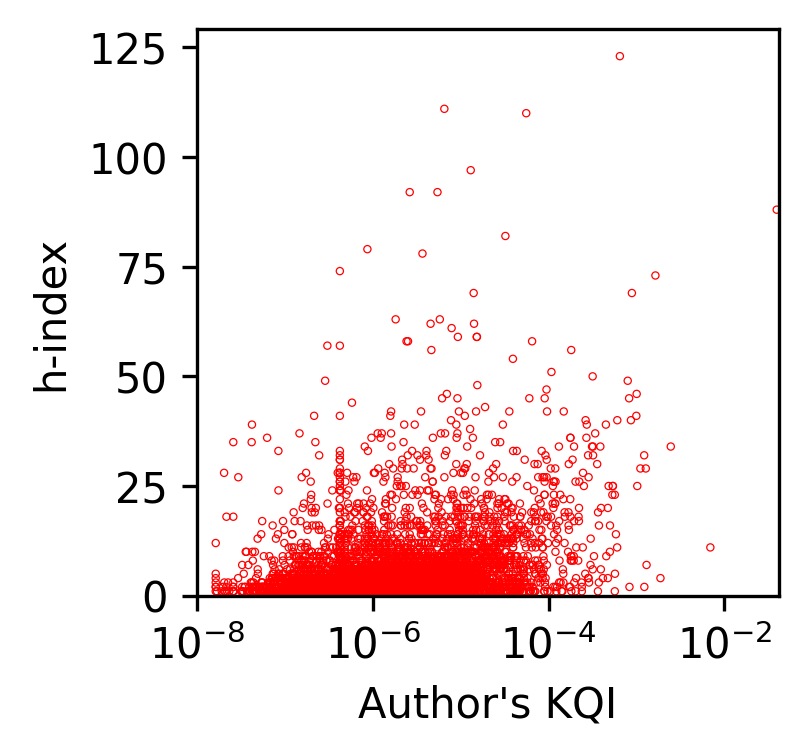}
    }
    \hspace{0.2cm}
    \subcaptionbox{Comparison with impact factor.}[4.5cm]{
        \includegraphics[height=4cm]{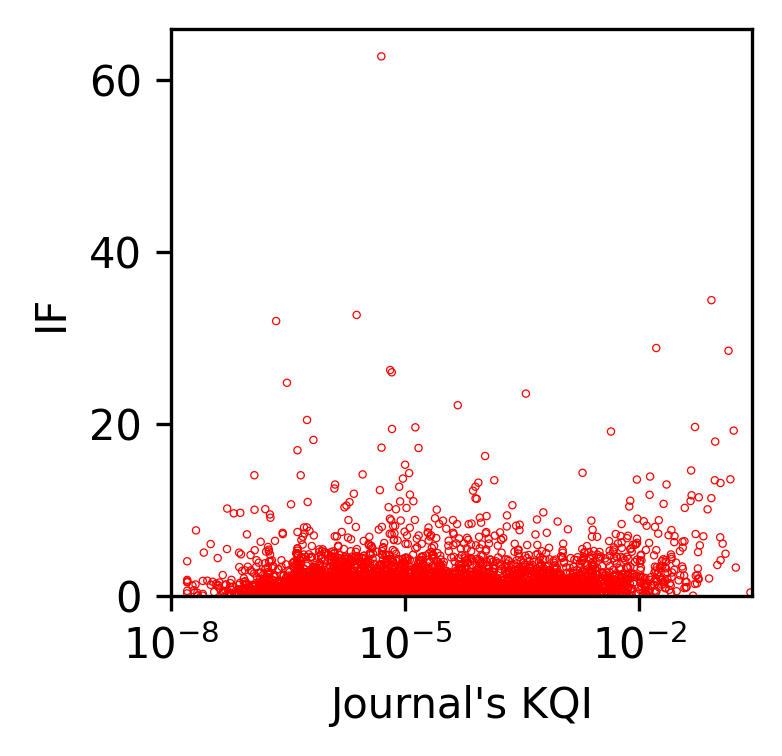}
    }
    \caption[Comparison among h-index, impact factor, PageRank and KQI.]{Comparison among h-index, impact factor, PageRank and KQI. KQI and PageRank show a significant positive correlation. The scholars with high h-index (citations) are more inclined to have higher KQI, and the impact factor of the journal has little relationship with its KQI.}
    \label{fig:comparison}
\end{figure}

% !TEX root = ../main.tex

\chapter{Case Study: Examples on Three Disciplines}
\label{chapter:casestudy}

In this chapter, we will take the fields of channel capacity and deep learning as examples to do a case study. The specific content includes the relationship between extraordinary literature and ordinary literature, the rise and fall of the disciplines reflected by KQI, and the knowledge veins diagram of the two fields.

\section{Extraordinary Literature set off by Ordinary Literature}

Fig. \ref{fig:redgreen} shows the KQI of \textit{The capacity of wireless networks} and changes in the field of channel capacity. \textit{The capacity of wireless networks} was born in 2000. In earlier years, it also served as a foil to other articles with high KQI. Later, it was set off by other ordinary articles and gradually became extraordinary. It is considered an extraordinary article because it is supported by millions of ordinary articles. Low KQI does not completely suggest that articles are not good. They are just the heroes behind extraordinary ones.

\begin{figure}[!htp]
    \centering
    \includegraphics[height=8cm]{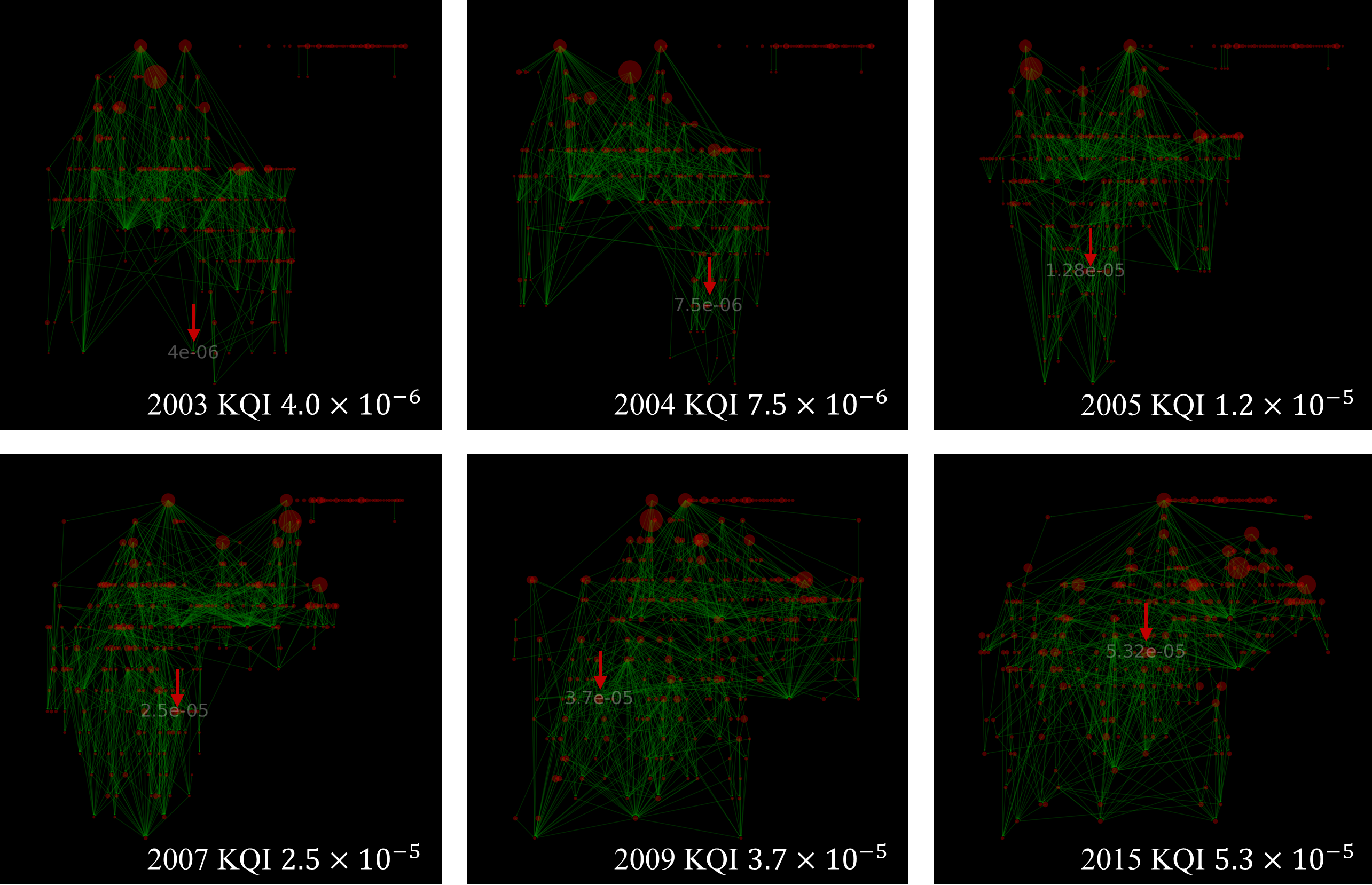} \\
    \caption{A tour with \textit{The capacity of wireless networks}.}
    \label{fig:redgreen}
\end{figure}

\section{Rise and Fall of Disciplines Reflected by KQI}

Fig. \ref{fig:casestudy_trend} shows that, the field of deep learning began to develop in 1989, then entered a flat period, and reached its peak around 2016. However, the field of channel capacity began to develop from \textit{Communication in the presence of noise} by Shannon in 1949, then boom in 2000 because of \textit{The capacity of wireless networks}, and finally stagnated in recent years.

\begin{figure}[!htp]
    \centering
    \includegraphics[height=8cm]{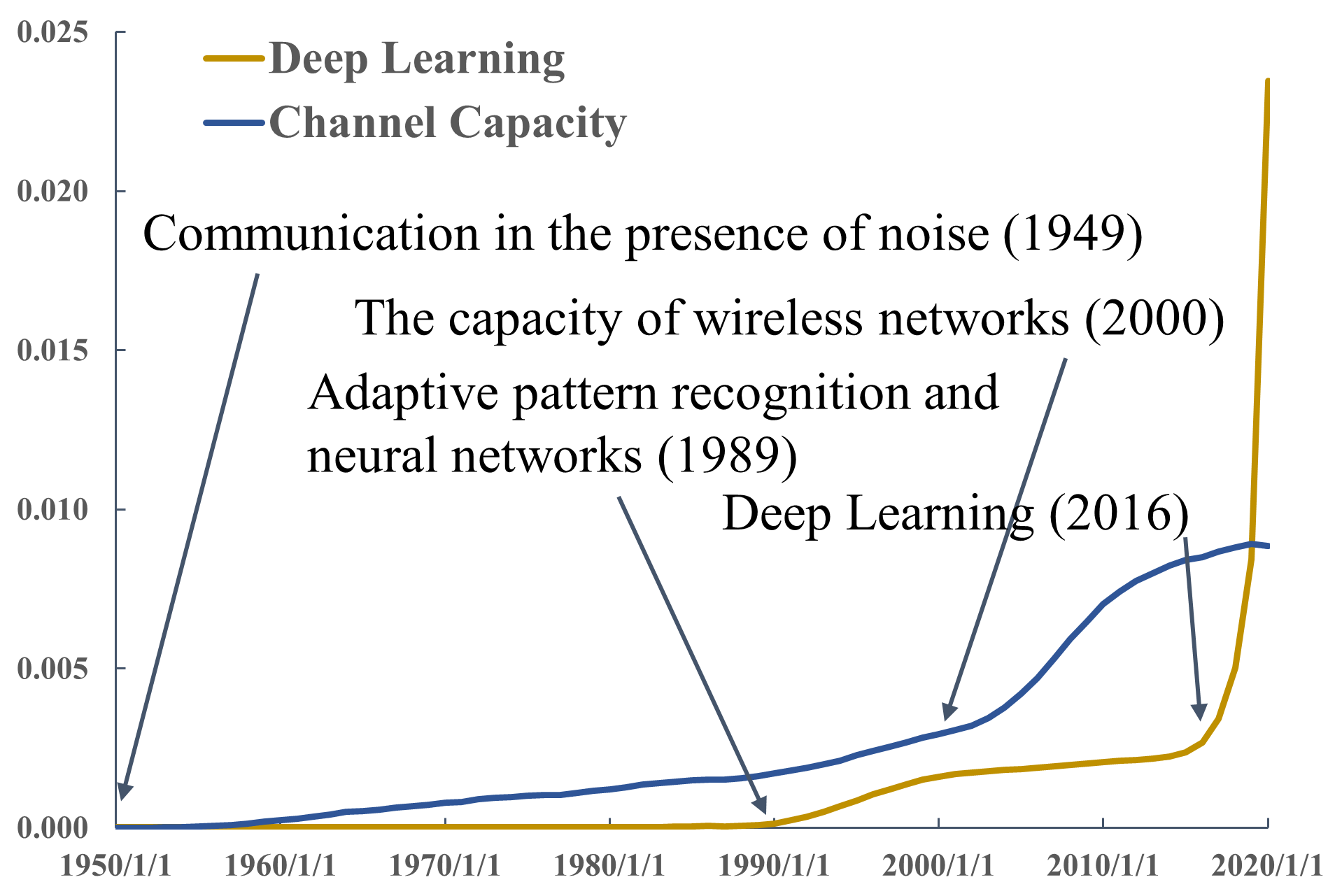} \\
    \caption{KQI trends of channel capacity and deep learning.}
    \label{fig:casestudy_trend}
\end{figure}

\section{Influential Authors with Top KQI}

Table \ref{tab:geoscience} shows the top 20 authors in field of geoscience. Most of them have great contributions and receive many awards, like members of NAS, fellows of AMS, etc.

\begin{table}[!hpt]
  \caption{Top@20 authors in geoscience.}
  \label{tab:geoscience}
  \centering
  \begin{tabular}{rcrl} 
    \toprule
    No & Author & KQI & Remarks \\
    \midrule
    1 & Syukuro Manabe & 0.00102 & Member of the National Academy of Sciences \\
    2 & Kevin E. Trenberth & 0.00072 & Fellow of the American Meteorological Society \\
    3 & Robert A. Berner & 0.00059 & Member of the National Academy of Sciences \\
    4 & Harold Jeffreys & 0.00045 & Fellow of the Royal Society\\
    5 & Dan McKenzie & 0.00044 & Fellow of the Royal Society \\
    6 & Richard W. Reynolds & 0.00043 & \\
    7 & Minze Stuiver & 0.00039 & Geological Society of America's Penrose Medal \\
    8 & J. Smagorinsky & 0.00038 & Known for General Circulation Mode \\
    9 & Philip D Jones & 0.00037 & Fellow of the American Geophysical Union\\
    10 & Lucia Solórzano & 0.00037 & \\
    11 & C. F. Ropelewski & 0.00036 & \\
    12 & Wallace S. Broecker & 0.00035 & Known for global warming \\
    13 & John M. Wallace & 0.00035 & Symons Gold Medal of the Royal Meteorological Society\\
    14 & Thomas R. Karl & 0.00035 & Verner E. Suomi Award by the American Meteorological Society \\
    15 & J.E. Nash & 0.00034 & \\
    16 & Piers J. Sellers & 0.00034 & American Meteorological Society Houghton Award \\
    17 & Samuel Epstein & 0.00032 & Known for carbonate paleothermometry\\
    18 & James E. Hansen & 0.00030 & Known for Climate models\\
    19 & Martin J. Buerger & 0.00029 & Member of the National Academy of Sciences \\
    20 & Roland von Huene & 0.00029 & Gustav Steinmann Medal\\
    \bottomrule
  \end{tabular}
\end{table}

\section{Knowledge Veins Diagram}

Fig. \ref{fig:Channel_capacity_vein} shows that the field of channel capacity originated from the \textit{Communication in the presence of noise} proposed by Shannon in 1949, which aroused wide discussion. And there were a lot of important developments around the year 2000.

Fig. \ref{fig:Deep_learning_vein} shows that the field of deep learning can be traced back to \textit{Approximation by superpositions of a sigmoidal function} in 1989 and \textit{The wake-sleep algorithm for unsupervised neural networks} in 1995. After 2010, there were a lot of exciting works: ImageNet, CNN, dropout, and so on.

Fig. \ref{fig:geoscience_vein} shows that the field of geoscience appears in two stages of development. Around the 1950s, carbon isotope methods flourished, and around the 1990s, global climate analysis became a hot topic.

\begin{figure}[!htp]
    \centering
    \includegraphics[height=14.8cm]{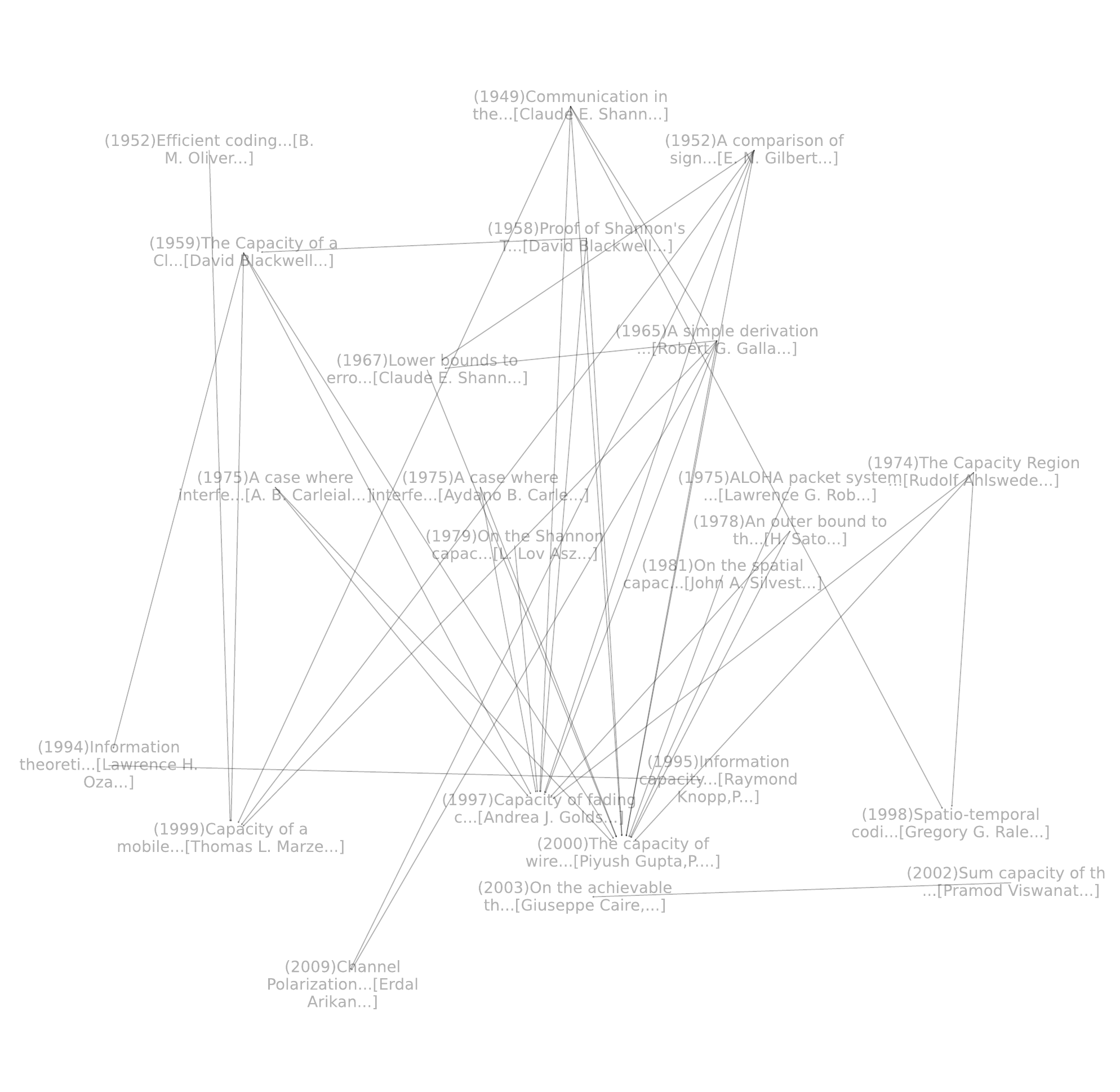} \\
    \caption{Knowledge vein of channel capacity.}
    \label{fig:Channel_capacity_vein}
\end{figure}

\begin{figure}[!htp]
    \centering
    \includegraphics[height=14.8cm]{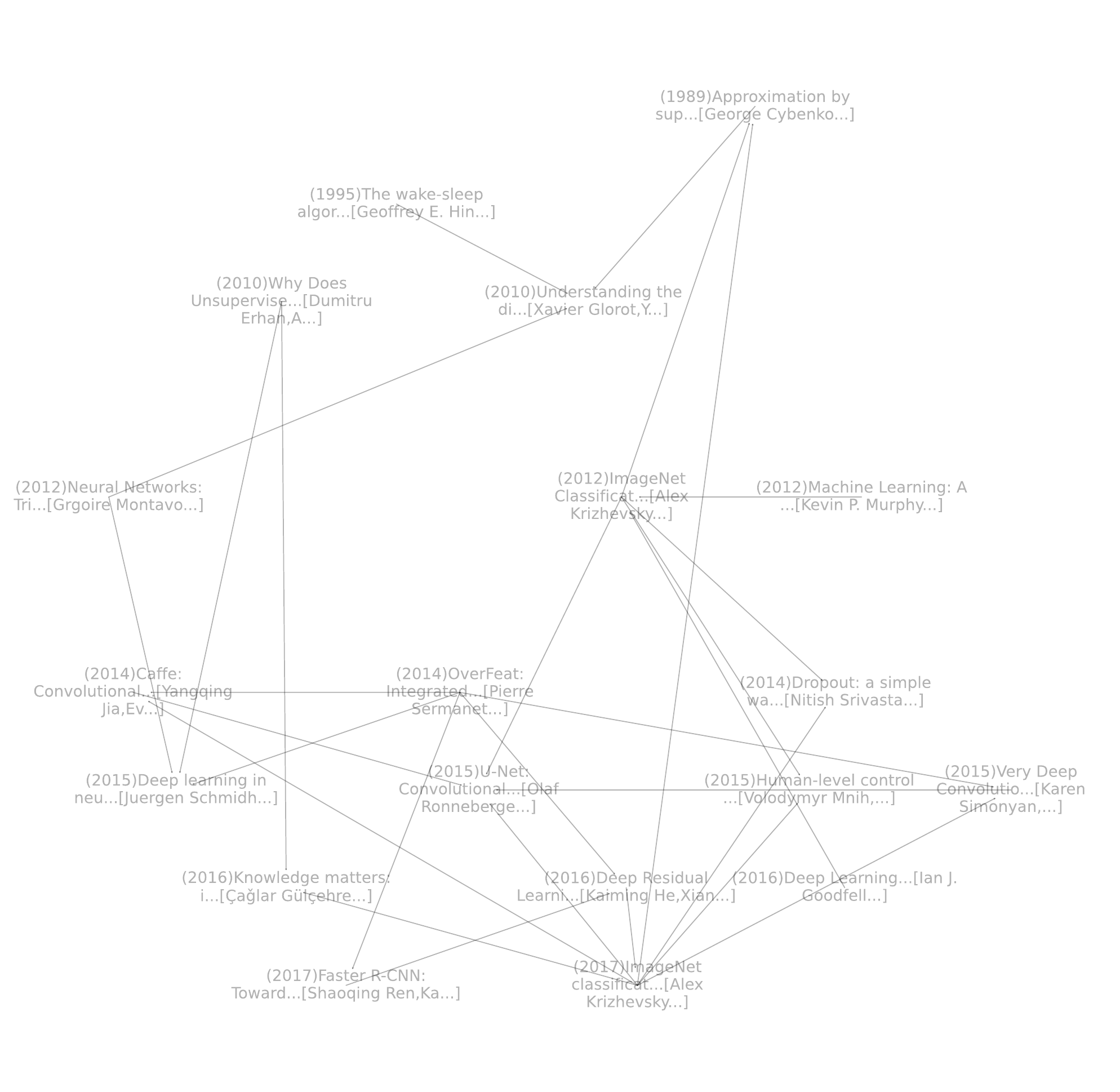} \\
    \caption{Knowledge vein of deep learning.}
    \label{fig:Deep_learning_vein}
\end{figure}

\begin{figure}[!htp]
    \centering
    \includegraphics[height=14.8cm]{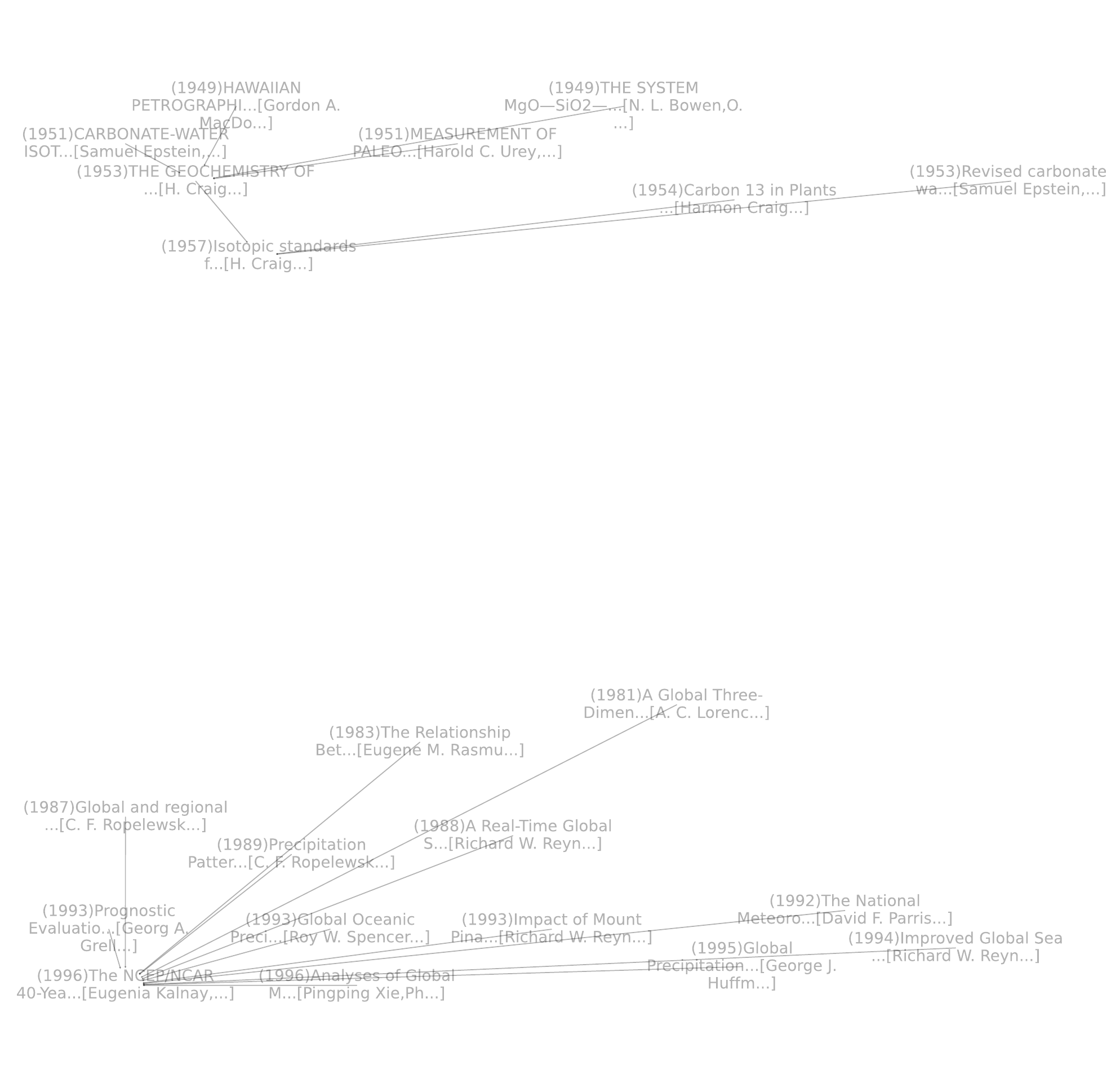} \\
    \caption{Knowledge vein of geoscience.}
    \label{fig:geoscience_vein}
\end{figure}

% !TEX root = ../main.tex

\begin{summary}

In this paper, we proposed KQI - a knowledge quantity index for evolving directed acyclic graph. The idea of KQI is to filter out structural characteristics from Shannon entropy and structure entropy. We use the Knowledge Tree as the partitioning tree to derive the formula and algorithm of KQI. A large number of experiments on academic dataset show that KQI obeys the linear growth law and 80/20 rule, and there is knowledge boom threshold on which knowledge growth begins to accelerate. Based on KQI, we also put forward the method of knowledge veins extraction, as well as the ranking of papers, authors and affiliations. Measurements in different disciplines show that KQI has significant effect in knowledge measurement, in comparison with other indicators.

In the process of academic development, the scientific productivity seems to present a miraculous accelerated growth, but of knowledge is a linear growth, when no accident. This accident refers to the acceleration in the growth of knowledge when a field of research has just reached the point of sufficient and mastery to be able to associate knowledge with ease. To some extent, this confirms the necessity of increasing government investment in scientific research, for sustainable development of knowledge. It also shows that we don't need to be scared by the explosion of knowledge, but just deal with the exhaustion of the explosion of information. Researchers can refer to KQI to select more promising disciplines or disciplines that are easier to start. They can also select papers with a higher KQI to quickly grasp the key concepts and methods of a new discipline, which will be more conducive to interdisciplinary research.

KQI is only a preliminary exploration of knowledge quantification, but it has shown exciting effects and has a broad exploration prospect in the future. For example,
\begin{itemize}
    \item Due to the arbitrariness of academic citation brought about by author relationship\cite{authorship}, evaluation indexes are often attacked by some scholars, such as self-citing and quoting irrelevant papers. Although KQI can do much better than citations, a thorough solution to this problem requires further consideration of the difference in the importance of references.
    \item It can be seen that the measurement of knowledge quantity in this paper is only from the perspective of structure, and we are also considering incorporating semantics to achieve better measurement effect. 
    \item Noting the remarkable effect of KQI in the extraction of knowledge veins, and we are also trying to construct the history of human academic development context based on this.
    \item At present, KQI only considers the volume of knowledge. In the future, it can continue to consider the value of knowledge into the indicators to make KQI more perfect.
\end{itemize}

According to the results presented in this article, perhaps scientific research will be accessible to everyone in the future, with academic inflation always taking place for knowledge. Despite the explosion in scientific productivity, calm down, we are only walking on the trail of exploring knowledge.

\end{summary}

%TC:ignore

% 使用英文字母对附录编号
% \appendix

% 附录内容
% \input{contents/app_maxwell_equations}
% \input{contents/app_flow_chart}

% 后文部分无编号
\backmatter

% 参考文献
\printbibliography[heading=bibintoc]

% 用于盲审的论文需隐去致谢、发表论文、科研成果、简历

% 致谢
% \input{contents/acknowledgements}

% 发表论文、科研成果、简历
% 盲审论文中，发表论文及科研成果等仅以第几作者注明即可，不要出现作者或他人姓名
% \input{contents/publications}
% \input{contents/achievements}
% \input{contents/resume}

% 中文学士学位论文要求在最后有一个英文大摘要，单独编页码，英文学士学位论文不需要
% \input{contents/english_digest}

%TC:endignore

\end{document}